\documentclass[a4paper,USenglish]{lipics-v2016}
 
\usepackage{microtype}
\usepackage{amssymb}
\usepackage{stmaryrd}
\usepackage{complexity}
\usepackage{ifthen}
\usepackage{multirow}%
\usepackage[table]{xcolor}
\usepackage{tikz}
\usetikzlibrary{calc}
\tikzstyle{every node}=[font=\scriptsize]

\def\centerarc[#1](#2)(#3:#4:#5)
{ \draw[line width=0.6pt,#1] ($(#2)+({#5*cos(#3)},{#5*sin(#3)})$) arc (#3:#4:#5); }

\def\carc(#1:#2:#3:#4:#5)
{ \centerarc[](#1,#2)(#3-(#4/2):#3+(#4/2):#5); }

\def\carcred(#1:#2:#3:#4:#5)
{ \centerarc[color=red](#1,#2)(#3-(#4/2):#3+(#4/2):#5); }

\def\carcblue(#1:#2:#3:#4:#5)
{ \centerarc[color=blue](#1,#2)(#3-(#4/2):#3+(#4/2):#5); }

\def\carcpurple(#1:#2:#3:#4:#5)
{ \centerarc[color=purple](#1,#2)(#3-(#4/2):#3+(#4/2):#5); }

\def\carcgreen(#1:#2:#3:#4:#5)
{ \centerarc[color=green](#1,#2)(#3-(#4/2):#3+(#4/2):#5); }

\def\arclabel(#1:#2:#3:#4:#5:#6)
{ \node (#1) at ($({#3+#5*cos(#6)},{#4+#5*sin(#6)})$) {#2};   }

\def\line(#1:#2:#3)
{ \draw[line width=0.6pt] (#1,#2) -- (#1+#3,#2); }

\def\vlinetikz[#1](#2:#3:#4)
{ \draw[line width=0.6pt,#1] (#2,#3) -- (#2,#3+#4); }

\def\beam[#1](#2:#3:#4:#5:#6)
{ \draw[line width=0.6pt,#1] ({#5*cos(#6)+#2},{#5*sin(#6)+#3}) -- ({#4*cos(#6)+#2},{#4*sin(#6)+#3});  }

\def\line(#1:#2:#3:#4)
{ \fill [white] (#1,#2) rectangle (#1+#3,#2+#4); }

\theoremstyle{plain}
\newtheorem{fact}[theorem]{Fact}

\newcommand{\N}{\mathbb{N}}

\newcommand{\simorbit}{\sim_{\mathrm{orb}}}
\newcommand{\canon}{\mathrm{canon}}

\newcommand{\set}[2]{\ensuremath{\left\{ #1 \mid #2 \right\}}}

\newcommand\restr[2]{{
        \left.\kern-\nulldelimiterspace 
        #1 
        \vphantom{\big|} 
        \right|_{#2} 
    }}

\newcommand{\circlecover}{\ensuremath{\mathtt{cc}}}
\newcommand{\overlap}{\ensuremath{\mathtt{ov}}}
\newcommand{\disjoint}{\ensuremath{\mathtt{di}}}
\newcommand{\contains}{\ensuremath{\mathtt{cs}}}
\newcommand{\contained}{\ensuremath{\mathtt{cd}}}
\newcommand{\itypeset}{\ensuremath{\{ \circlecover, \contained, \contains, \disjoint, \overlap  \}}}

\newcommand{\catype}{\ensuremath{\mathrm{type}}}

\newcommand{\ovtit}{\ensuremath{\mathcal{T}_{\mathrm{IT}}}}
\newcommand{\ovtnht}{\ensuremath{\mathcal{T}_{\mathrm{NHT}}}}

\newcommand{\Funiform}{F_{\mathrm{uniform}}}

\title{Canonical Representations for Circular-Arc Graphs Using Flip Sets}

\author{Maurice Chandoo}

\affil{Leibniz Universität Hannover, Theoretical Computer Science,\\
    Appelstr. 4, 30167 Hannover, Germany\\
    \texttt{chandoo@thi.uni-hannover.de}}
\authorrunning{M. Chandoo} 

\Copyright{Maurice Chandoo}


\begin{document}

\maketitle

\begin{abstract}
	We show that computing canonical representations for circular-arc (CA) graphs reduces to computing certain subsets of vertices called flip sets. 
	For a broad class of CA graphs, which we call uniform, it suffices to compute a CA representation to find such flip sets. 
	As a consequence canonical representations for uniform CA graphs can be obtained in polynomial-time. 
    We then investigate what kind of CA graphs pose a challenge to this approach. This leads us to introduce the notion of restricted CA matrices and show that the canonical representation problem for CA graphs is logspace-reducible to that of restricted CA matrices. As a byproduct, we obtain the result that CA graphs without induced 4-cycles can be canonized in logspace. 
\end{abstract} 

\section{Introduction}
We consider an arc to be a connected set of points on the unit circle including the endpoints. 
A CA graph is a graph whose vertices can be assigned arcs such that two vertices are adjacent iff their corresponding arcs intersect. More formally, given a graph $G$ we call it a CA graph if there exists a function $\rho$ which maps every vertex $u$ of $G$ to an arc $\rho(u)$ such that $u$ and $v$ are adjacent iff their arcs $\rho(u)$ and $\rho(v)$ have non-empty intersection. We call such a mapping $\rho$ a CA representation of $G$. CA graphs are a form of geometrical intersection graphs. Let $\mathcal{X}$ be a family of sets over some ground set. Any subset $Y$ of $\mathcal{X}$ defines a graph $G_Y$ which has $Y$ as its vertex set and two vertices are adjacent if they have non-empty intersection. The graph $G_Y$ is called intersection graph of $Y$. We say a (finite) graph $G$ is an intersection graph of $\mathcal{X}$ if it is isomorphic to the intersection graph of $Y$ for some $Y \subseteq \mathcal{X}$. 
In this language CA graphs are intersection graphs of arcs. The intersection graphs of intervals on the real line are called interval graphs. In this sense any set of geometrical objects defines a (geometrical intersection) graph class. CA graphs are a generalization of interval graphs since every set of intervals on the real line can be `bent' into arcs while preserving the intersection relation. Therefore every interval graph is a CA graph.    

Being a generalization of interval graphs---the archetype of geometrical intersection graphs---CA graphs are quite prominent as well and have been known for decades.
Since then structural properties and algorithmic problems for this class have been thoroughly investigated with \cite{gav} and \cite{tuc} being two of the earliest works in this regard. In particular, finding characterizations of CA graphs and constructing a CA representation for a given CA graph have received a great deal of attention. Remarkably, finding a forbidden induced subgraph characterization of CA graphs is still an open problem. See \cite{lin2} for a survey on this line of research and \cite{cao} for one of the most recent results in that direction. It should also be mentioned that CA graphs are of practical relevance with applications arising in disciplines such as genetics and operations research. An explanation of the connection between genetics and interval graphs in layman's terms can be found in \cite{wat}. For a specialized account on this connection emphasizing circularity see \cite{stahl}. An example of how CA graphs can be used to model the problem of phasing traffic lights is given in \cite{gol}. 

In this work we consider the canonical representation problem for CA graphs. The representation problem for CA graphs is as follows. Given a CA graph $G$ as input we want to output a CA representation $\rho_G$ of $G$. The canonical variant of this problem imposes the additional requirement that for every pair of isomorphic CA graphs $G$ and $H$ their representations $\rho_G$ and $\rho_H$ should have identical underlying sets of arcs, i.e.~$\set{\rho_G(v)}{ v \in V(G)} = \set{\rho_H(v)}{ v \in V(H)}$. Notice that solving the representation problem for CA graphs implies solving the recognition problem for CA graphs, i.e.~the question given a graph $G$ is it a CA graph. Likewise, solving the canonical representation problem for CA graphs implies solving the isomorphism problem for CA graphs, i.e.~deciding whether two given CA graphs are isomorphic. 

Consider the following generalization of interval graphs: 2-interval graphs are intersection graphs of two intervals on the real line. It is easy to see that this class contains CA graphs because given a set of arcs one can cut the circle at some point and straighten the arcs. The arcs which are cut can be modeled as two intervals. It is interesting to note that the isomorphism problem for interval graphs is logspace-complete \cite{kob:intv} while the one for 2-interval graphs is already GI-complete and CA graphs lie inbetween these two classes. The GI-completeness for 2-interval graphs follows from the fact that every line graph is a 2-interval graph and line graphs are already GI-complete. To see why this inclusion holds consider a graph $G$ and its line graph $L(G)$. Assign every vertex $v$ of $G$ an interval $I_v$ on the real line such that no two intervals $I_u$ and $I_v$ intersect for every pair of distinct vertices $u$ and $v$ of $G$. The 2-interval model for $L(G)$ is obtained by mapping every edge $\{u,v\}$ of $G$ to the two intervals $I_u$ and $I_v$. 

While a polynomial-time algorithm for deciding isomorphism of interval graphs is known since 1976 due to Booth and Lueker\nocite{lue} this question still remains open for CA graphs.    
There have been two claimed polynomial-time algorithms for deciding isomorphism of CA graphs in \cite{wu} and \cite{hsu} which were shown to be incorrect in \cite{esc} and \cite{cur} respectively. For interval graphs even a linear-time algorithm for isomorphism is known \cite{lue}. A more recent result is that canonical interval representations for interval graphs can be computed in logspace and that this is optimal in the sense that recognition and deciding isomorphism for interval graphs is logspace-complete \cite{kob:intv}. These two hardness results also carry over to the class of CA graphs.  Furthermore, the isomorphism problem for proper CA graphs \cite{kob:pca} and Helly CA graphs \cite{kob:hca} have been shown to be decidable in logspace. It is also shown how to obtain canonical representations for these subclasses in logspace.

In this article we explain how the method used in \cite{kob:hca} to obtain canonical representation for Helly CA graphs can be adapted to CA graphs in general. Following this approach, canonical representations for CA graphs can be found by computing certain subsets of vertices called flip sets in an isomorphism-invariant manner. We introduce the class of uniform CA graphs for which this method yields canonical representations in polynomial-time. We then aim to isolate the instances of CA graphs which are difficult to handle with this method. We try to capture these hard instances by what we call restricted CA matrices and show that the canonical representation problem for CA graphs is logspace-reducible to that of restricted CA matrices.
During this isolation process we find a subset of uniform CA graphs, namely $\Delta$-uniform CA graphs, for which canonical representations can be computed in logspace. The $\Delta$-uniform CA graphs contain Helly CA graphs and every CA graph without an induced 4-cycle.
This generalizes the canonization result for Helly CA graphs given in \cite{kob:hca}.
A preliminary version of this work appeared in \cite{cha1}.

The paper is organized as follows. In the third section we formalize the idea of computing invariant flip sets in order to obtain canonical representations for CA graphs. This leads to the definition of invariant flip set functions. In the fourth section we investigate for what CA graphs a particular invariant flip set function is easy to compute. This leads to the class of uniform CA graphs. We also provide an alternative characterization of uniform CA graphs in terms of whether certain triangles in a CA graph have an unambiguous representation. The main result of this section is that the representation problem for uniform CA graphs, the canonical representation problem for uniform CA graphs and the non-Helly triangle representability problem (introduced in section 4) for uniform CA graphs are all logspace-equivalent. In the fifth section we consider the structure of non-uniform CA graphs, introduce restricted CA matrices and show how the canonical representation problem for CA graphs can be reduced to that of restricted CA matrices. In the process of proving this reduction the class of $\Delta$-uniform CA graphs is defined and it is shown that this class can be canonized in logspace.

\section{Preliminaries}
    For a number $n \in \mathbb{N}$ we write $[n]$ for $\{1,\dots,n \}$.
    Given two sets $A,B$ we say $A$ and $B$ intersect if $A \cap B \neq \emptyset$. We say $A$ and $B$ overlap, in symbols $A \between B$, if $A \cap B, A \setminus B$ and $B \setminus A$ are non-empty.     
    
    We consider graphs without self-loops which sometimes have colored vertices and colored edges. They can be seen as relational structures with the vertex set as universe and vertex colors encoded as unary relations and edge colors as binary relations. The standard notion of isomorphism for relational structures applies. We describe a graph with vertex colors as tuple $(G,c)$ where $c$ is a function that maps the vertices of $G$ to the colors. We talk about a graph with edge colors as a square matrix whose entries represent the edge colors and identify the indices of the matrix and the vertices of the graph. Consequently, we identify a square matrix with the graph that it represents and talk about it in graph-theoretical terms. 
    By a class of (relational) structures we mean a set of such structures which is closed under isomorphism.    
       
    We call a bijective function $\tau$ which maps the vertices of a graph $G$ to some set $V'$ a relabeling of $G$ and $\tau(G)$ denotes the graph obtained after relabeling the vertices of $G$ according to $\tau$. 
    Let $G$ and $H$ be two graphs and let $X \subseteq V(G)$ and $Y \subseteq V(H)$. We say $X$ and $Y$ are in the same orbit, in symbols $X \simorbit Y$,  if there exists an isomorphism $\pi$ from $G$ to $H$ such that $\pi(X) = Y$. Let $f$ be a function which maps a graph along with a subset of its vertex set to a binary string, i.e.~$f(G,X) \in \{0,1\}^*$ and $X \subseteq V(G)$. We call $f$ an invariant for a graph class $\mathcal{C}$ if $f(G,X) = f(H,Y)$ whenever $X \simorbit Y$ and $G,H \in \mathcal{C}$.        
    Let us call a function $f$ which maps a graph $G$ to a family of subsets of its vertex set, i.e.~$f(G) \subseteq \mathcal{P}(V(G))$, a vertex set selector.
    For example, the function that maps a graph to the set of its cliques is a vertex set selector.
    The characteristic function $\chi_f$ of a vertex set selector $f$ is defined as $\chi_f(G,X) = 1 \Leftrightarrow X \in f(G)$.     
    We say a vertex set selector $f$ is invariant for a graph class $\mathcal{C}$ if its characteristic function $\chi_f$ is an invariant for $\mathcal{C}$. We call $f$ globally invariant if $\chi_f$ is an invariant for all graphs. 
    Intuitively, a vertex set selector $f$ is invariant for $\mathcal{C}$ if a graph $G \in \mathcal{C}$ can be arbitrarily relabeled and $f$ still returns the `same' vertex sets as before w.r.t.~$\simorbit$.
    
    The following definitions are with respect to a graph $G$. Throughout the paper it will be always clear from context with respect to what graph these expressions are to be interpreted. 
    For a vertex $v$ we define its open neighborhood $N(v)$ as the set of vertices which are adjacent to $v$ and its closed neighborhood $N[v] = N(v) \cup \{v\}$. 
    A vertex $v$ is called universal if $N[v] = V(G)$. For two vertices $u, v$ we say that $u$ and $v$ are twins if $N[u] = N[v]$. A graph $G$ is twin-free if for every pair of distinct vertices $u \neq v$ it holds that $N[u] \neq N[v]$. A twin class is an inclusion-maximal set of vertices $X$ such that for all $u,v \in X$ it holds that $u$ and $v$ are twins.
    For two subsets of vertices $S,S'$ with $S' \subseteq S$ we define the exclusive neighborhood $N_{S}(S')$ 
    as all vertices $v \in V(G) \setminus S$ such that $v$ is connected to all vertices in $S'$ and to none in $S \setminus S'$. Let $A$ be a square matrix with entries from a set $\mathcal{E}$. For a vertex $u$ of $A$ and $x \in \mathcal{E}$ we define $N^x(u) = \set{v \in V}{ A_{u,v} = x}$.
    
    \subparagraph*{Logspace Transducers and Reductions.} 
    We assume deterministic Turing machines as default model of computation.
    A logspace transducer is a deterministic Turing machine $M$ with a read-only input tape, a work tape and a write-only output tape. The work tape is only allowed to use at most $\mathcal{O}(\log n)$ cells where $n$ denotes the input length. 
    To write onto the output tape $M$ has a designated state called output state with the following semantic. If $M$ enters the output state then the symbol in the current cell of the work tape is written to the current cell of the output tape and the head on the output tape is moved one cell to the right. Other than that, $M$ cannot write or move the head on the output tape.     
    This means as soon as something is written to the output tape it cannot be modified afterwards.     
    Let $\Sigma$ and $\Gamma$ be the input and work alphabet of $M$ respectively. Then $M$ computes a function $f_M \colon \Sigma^* \rightarrow \Gamma^*$.     
    We say a (partial) function $f$ is computed by a logspace transducer $M$ if $f(x) = f_M(x)$ whenever $f(x)$ is defined. We call $f$ logspace-computable if there exists a logspace transducer $M$ which computes $f$. The class of logspace-computable functions is closed under composition. 
    Let $f$ be a function which maps words over some alphabet to words over some other alphabet. We say that the length of $f$ is polynomially bounded if $|f(x)|$ is polynomially bounded by $|x|$. Only functions whose length is polynomially bounded can be logspace-computable since the runtime of a logspace transducer is polynomially bounded.    
    A language $L \subseteq \Sigma^*$ is in logspace if its characteristic function is logspace-computable.
    
    Given two functions $f$ and $g$ we say $f$ is logspace-reducible to $g$ if there exists $l \in \N$ and logspace-computable functions $r_1,\dots,r_l$ such that $f$ can be expressed as composition of $g$ and $r_1,\dots,r_l$. Intuitively, this means that an oracle which can compute $g$ can be queried a constant number of times when constructing a logspace transducer for $f$.  
    For two functions $f$ and $g$ we say that they are logspace-equivalent if $f$ is logspace-reducible to $g$ and vice versa. Analogously, given three functions $f,g_1,g_2$ we say $f$ is logspace-reducible to $g_1$ and $g_2$ if there exists $l \in \N$ and logspace-computable functions $r_1,\dots,r_l$ such that $f$ can be expressed as composition of $g_1,g_2$ and $r_1,\dots,r_l$.                
    
    \subparagraph*{Circular-Arc Graphs and Representations.}
    A CA model is a set of arcs $\mathcal{A} = \{A_1,\dots,A_n\}$ on the circle. Let $p \neq p'$ be two points on the circle. Then the arc $A$ specified by $[p,p']$ is given by the part of the circle that is traversed when starting from $p$ going in clockwise direction until $p'$ is reached. We say that $p$ is the left and $p'$ the right endpoint of $A$ and write $l(\cdot),r(\cdot)$ to denote the left and right endpoint of an arc in general. If $A = [p,p']$ then the arc obtained by swapping the endpoints $\overline{A} = [p',p]$ covers exactly the opposite part of the circle plus the endpoints. We say $\overline{A}$ is obtained by flipping $A$.
    In our context, we are only interested in the intersection structure of a CA model and thus only the relative position of the endpoints to each other matter. All endpoints can w.l.o.g.~be assumed to be pairwise different and no arc covers the full circle. Under these assumptions, a CA model $\mathcal{A}$ with $n$ arcs can be described as a unique string as follows.      
    Pick an arbitrary arc $A \in \mathcal{A}$ and relabel the arcs with $1, \dots, n$ in order of appearance of their left endpoints when traversing the circle clockwise starting from the left endpoint of $A$. Then write down the endpoints in order of appearance when traversing the circle  clockwise starting from the left endpoint of $A$. Do this for every arc and pick the lexicographically smallest resulting string as representation for $\mathcal{A}$. For example, the smallest such string for the CA model in Figure~\ref{fig:ca_intro} would result from choosing $A_1$: ($l(1),r(1),l(2),r(5),l(3),r(2),\dots$). Let $\mathrm{str}(\mathcal{A})$ denote this smallest string representation. 
    For a CA model $\mathcal{A}$ let $\mathcal{A}^{\mathrm{r}}$ be the CA model obtained after reversing the order of its endpoints.
    Observe that reversing the endpoints does not affect the intersection structure of a CA model. Therefore we consider two CA models $\mathcal{A}$ and $\mathcal{B}$ to be equal if $\mathrm{str}(\mathcal{A}) = \mathrm{str}(\mathcal{B})$ or $\mathrm{str}(\mathcal{A}^{\mathrm{r}}) = \mathrm{str}(\mathcal{B})$.
    
    Let $G$ be a graph and $\rho = (\mathcal{A},f)$ consists of a CA model $\mathcal{A}$ and a bijective mapping $f$ from the vertices of $G$ to the arcs in $\mathcal{A}$. Then $\rho$ is called a CA representation of $G$ if for all $u \neq v \in V(G)$ it holds that $\{u,v\} \in E(G) \Leftrightarrow f(u) \cap f(v) \neq \emptyset$. We write $\rho(x)$ to mean the arc $f(x)$ corresponding to the vertex $x$, $\rho(G)$ for the CA model $\mathcal{A}$ and for a subset $V' \subseteq V(G)$ let $\rho[V'] = \left\{ \rho(v) \mid v \in V' \right\}$.
    A graph is a CA graph if it has a CA representation. 
    
    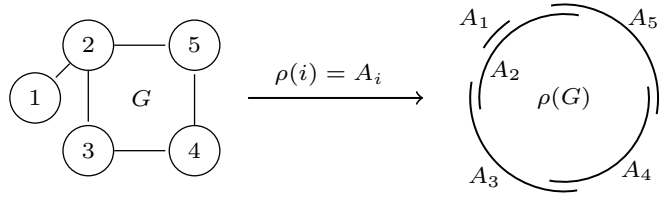
\begin{figure}
    \centering
    \resizebox{9cm}{!}{%
        \begin{tikzpicture}[shorten >=1pt,auto,node distance=1.2cm,
  main node/.style={circle,draw}]

\newcommand*{\xoff}{0}%
\newcommand*{\yoff}{0}%

\newcommand*{\xoffca}{4.5}%
\newcommand*{\yoffca}{0}%

\newcommand*{\movea}{1.4}%
\newcommand*{\moveb}{0.8*sqrt(1/2)}%

\node[main node] (v1) at ({\xoff-\moveb},{\yoff+\moveb}) {2};
\node[main node] (v2) at ({\xoff-\moveb},{\yoff-\moveb}) {3};
\node[main node] (v3) at ({\xoff+\moveb},{\yoff-\moveb}) {4};
\node[main node] (v4) at ({\xoff+\moveb},{\yoff+\moveb}) {5};
\node[main node] (v5) at ({\xoff-2*\moveb},{\yoff}) {1};
\path[-]
(v1) edge (v4)
(v1) edge (v2)
(v1) edge (v5)
(v2) edge (v3)
(v3) edge (v4)
;
\node (M) at (\xoff,\yoff) {$G$};

\draw[line width=0.6pt,<-] (\xoff+3,\yoff) -- (\xoffca-3.4,\yoff);
\node (M) at (\xoff+2,\yoff+0.25) {$\rho(i) = A_i$};

\centerarc[](\xoffca,\yoffca)(90-10:180+10:0.9);
\centerarc[](\xoffca,\yoffca)(180-10:270+10:1);
\centerarc[](\xoffca,\yoffca)(270-10:360+10:0.9);
\centerarc[](\xoffca,\yoffca)(0-10:90+10:1);
\centerarc[](\xoffca,\yoffca)(125:150:1.01);

\node (a1) at ($({\xoffca+1.2*cos(45)},{\yoffca+1.2*sin(45)})$) {$A_5$};  
\node (a2) at ($({\xoffca+1.1*cos(45-90)},{\yoffca+1.1*sin(45-90)})$) {$A_4$};  
\node (a3) at ($({\xoffca+1.2*cos(45-180)},{\yoffca+1.2*sin(45-180)})$) {$A_3$};   
\node (a4) at ($({\xoffca+0.75*cos(-195)+0.1},{\yoffca+0.75*sin(-195)+0.1})$) {$A_2$};      
\node (a5) at ($({\xoffca+1.2*cos(45+90)-0.1},{\yoffca+1.2*sin(45+90)})$) {$A_1$};      
\node (M) at (\xoffca,\yoffca) {$\rho(G)$};

\end{tikzpicture}
    }
    \caption{A CA graph and a representation of it}
    \label{fig:ca_intro}                
    \end{figure}          
    
    We say a CA model $\mathcal{A}$ has a hole if there exists a point on the circle which isn't contained by any arc in $\mathcal{A}$. Every such CA model can be understood as interval model (a set of intervals on the real line)  by straightening the arcs. Conversely, every interval model can be seen as CA model by bending the intervals. Therefore a graph is an interval graph iff it admits a CA representation with a hole.
    
    A family of sets $\mathcal{F}$ over some ground set is called Helly if for all subsets $\mathcal{F}'$ of $\mathcal{F}$ such that all elements in $\mathcal{F}'$ intersect pairwise it holds that $\cap_{A \in \mathcal{F}'} A$ is non-empty.    
    A CA graph $G$ is called Helly (HCA graph) if it has a CA representation $\rho$ with a Helly CA model $\rho(G)$. This is the case iff for all inclusion-maximal cliques $C$ in $G$ it holds that the overall intersection of $C$ in $\rho$ is non-empty, i.e.~$\bigcap_{v \in C} \rho(v) \neq \emptyset$. Every interval model has the Helly property and therefore every interval graph is a Helly CA graph.         
        
    The intersection type of two circular arcs $A$ and $B$ can be one of the following five types:
        \begin{itemize}
            \item $\disjoint$: $A$ and $B$ are disjoint  --- $A \cap B = \emptyset$
            \item $\contains$: $A$ contains $B$  --- $B \subset A$
            \item $\contained$: $A$ is contained by $B$  --- $A \subset B$        
            \item $\circlecover$: $A$ and $B$ jointly cover the circle (circle cover) --- $A \between B$ and $A \cup B = $ whole circle
            \item $\overlap$: $A$ and $B$ overlap  --- $A \between B$ and $A \cup B \neq $ whole circle
        \end{itemize}       
    
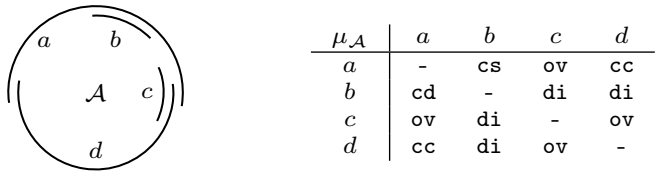
\begin{figure}[b]
    \centering
    \resizebox{9cm}{!}{%
        \begin{tikzpicture}[shorten >=1pt,auto,node distance=1.2cm,
  main node/.style={circle,draw}]

\newcommand*{\xoff}{0}%
\newcommand*{\yoff}{0}%

\newcommand*{\xoffca}{0}%
\newcommand*{\yoffca}{0}%

\carc(\xoffca:\yoffca:90:198:0.9);
\carc(\xoffca:\yoffca:270:198:0.8);
\carc(\xoffca:\yoffca:70:50:0.8);
\carc(\xoffca:\yoffca:0:50:0.7);

\arclabel(cam:$\mathcal{A}$:\xoffca:\yoffca:0:0);
\arclabel(a:$a$:\xoffca:\yoffca:0.75:135);
\arclabel(b:$b$:\xoffca:\yoffca:0.6:70);
\arclabel(c:$c$:\xoffca:\yoffca:0.53:0);
\arclabel(d:$d$:\xoffca:\yoffca:0.6:270);

\node (tab) at (4,0) { 
\begin{tabular}{ c |  c c c c }
  $\mu_{\mathcal{A}}$ & $a$ & $b$ & $c$ & $d$ \\
  \hline
$a$ & - & $\contains$ & $\overlap$ & $\circlecover$ \\
$b$ & $\contained$ & - & $\disjoint$ & $\disjoint$ \\
$c$ & $\overlap$ & $\disjoint$ & - & $\overlap$ \\
$d$ & $\circlecover$ & $\disjoint$ & $\overlap$ & - \\
\end{tabular}
};


\end{tikzpicture}
    }
    \caption{A CA model $\mathcal{A}$ and its intersection matrix $\mu_\mathcal{A}$}
    \label{fig:imatrix}  
\end{figure}         
    
    Using these types we can associate a matrix with every CA model. An intersection matrix is a square matrix with entries $\itypeset$. Given a CA model $\mathcal{A}$ we define its intersection matrix $\mu_\mathcal{A}$ such that $(\mu_{\mathcal{A}})_{A,B} \in \itypeset$ reflects the intersection type of the arcs $A \neq B \in \mathcal{A}$. An intersection matrix $\mu$ is called a CA (interval) matrix if it is the intersection matrix of some CA model (with a hole). See Figure \ref{fig:imatrix} for an example of a CA model and the CA matrix which it induces. 
    Given an intersection matrix $\mu$ and two distinct vertices $u, v$ of $\mu$ we sometimes write $u \: \alpha \: v$ instead of $\mu_{u,v} = \alpha$ if $\mu$ is clear from the context. Also, we sometimes talk about an intersection matrix $\mu$ as if it were a graph. In that case we consider two vertices $u,v$ of $\mu$ to be adjacent if they do not have a $\disjoint$-entry in $\mu$. 
        
    When trying to construct a CA representation for a CA graph $G$ it is clear that whenever two vertices are non-adjacent their corresponding arcs must be disjoint in every CA representation of $G$. For two adjacent vertices the intersection type of their corresponding arcs might depend on the particular CA representation of $G$ that one considers. Hsu has shown that this ambiguity can be removed as follows \cite{hsu}. 
    
    We adopt the notation of \cite{kob:hca}.
    
    \begin{definition}
        For a graph $G$ we define its neighborhood matrix $\lambda_G$ which is an intersection matrix as        
    \[
    (\lambda_G)_{u,v} = 
    \begin{cases}
    \disjoint &, \text{if } \{ u,v \} \notin E(G) \\
    \contained &, \text{if } N[u] \subsetneq N[v] \\
    \contains &, \text{if } N[v] \subsetneq N[u] \\
    \circlecover &, \text{if } N[u] \between N[v] \text{ and } N[u] \cup N[v] = V(G) \\
    & \text{ and } \forall w \in N[u]\setminus N[v]: N[w] \subset N[u] \\
    & \text{ and } \forall w \in N[v]\setminus N[u]: N[w] \subset N[v] \\
    \overlap &, \text{otherwise}
    \end{cases}
    \]
    for all $u \neq v \in V(G)$. 
    \end{definition} 

    Let $\mu$ be an intersection matrix with vertex set $V$ and let $\rho = (\mathcal{A},f)$ where $\mathcal{A}$ is a CA model and $f$ is a bijective mapping from $V$ to $\mathcal{A}$. We say $\rho$ is a CA representation of $\mu$ if $f$ is an isomorphism from $\mu$ to the intersection matrix $\mu_{\mathcal{A}}$ of $\mathcal{A}$. We denote the set of such CA representations for $\mu$ with $\mathcal{N}(\mu)$. The representation problem for CA matrices is to compute a CA representation for a given CA matrix $\mu$. The canonical representation problem for CA matrices is defined analogously to the canonical representation problem for CA graphs.    
    We say $\rho$ is a normalized CA representation for a graph $G$ if $\rho$ is a CA representation for the neighborhood matrix $\lambda_G$ of $G$. An example of a normalized representation can be seen in Figure~\ref{fig:normrepr}.  
    Let us denote the set of all normalized CA representations for $G$ with $\mathcal{N}(G) = \mathcal{N}(\lambda_G)$.    
    
    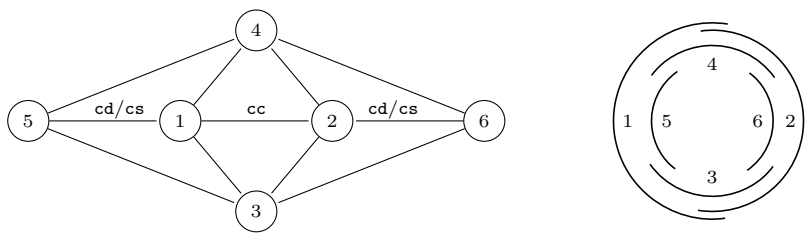
\begin{figure}[b]
        \centering
        \begin{tikzpicture}[shorten >=1pt,auto,node distance=1.2cm,
  main node/.style={circle,draw}]

\newcommand*{\xoff}{0}%
\newcommand*{\yoff}{0}%

\newcommand*{\xoffca}{6}%
\newcommand*{\yoffca}{0}%

\newcommand*{\movea}{1}%

\node (midcc) at ({\xoff},{\yoff+0.15}) {$\circlecover$};
\node (leftcd) at ({\xoff-1.8*\movea},{\yoff+0.15}) {$\contained$/$\contains$};
\node (rightcd) at ({\xoff+1.8*\movea},{\yoff+0.15}) {$\contained$/$\contains$};

\node[main node] (v1) at ({\xoff-\movea},{\yoff}) {$1$};
\node[main node] (v2) at ({\xoff+\movea},{\yoff}) {$2$};
\node[main node] (v3) at ({\xoff},{\yoff-1.2*\movea}) {$3$};
\node[main node] (v4) at ({\xoff},{\yoff+1.2*\movea}) {$4$};
\node[main node] (v5) at ({\xoff-3*\movea},{\yoff}) {$5$};
\node[main node] (v6) at ({\xoff+3*\movea},{\yoff}) {$6$};

\path[-]
(v1) edge (v3)
(v1) edge (v4)
(v2) edge (v3)
(v2) edge (v4)

(v1) edge (v2)

(v5) edge (v1)
(v5) edge (v3)
(v5) edge (v4)

(v6) edge (v2)
(v6) edge (v3)
(v6) edge (v4)
;

\carc(\xoffca:\yoffca:180:198:1.3);
\carc(\xoffca:\yoffca:0:198:1.2);

\carc(\xoffca:\yoffca:0:110:0.8);
\carc(\xoffca:\yoffca:180:110:0.8);
\carc(\xoffca:\yoffca:90:110:1);
\carc(\xoffca:\yoffca:270:110:1);

\arclabel(a1:$1$:\xoffca:\yoffca:1.11:180);
\arclabel(a2:$2$:\xoffca:\yoffca:1.03:0);

\arclabel(a5:$5$:\xoffca:\yoffca:0.6:180);
\arclabel(a6:$6$:\xoffca:\yoffca:0.6:0);

\arclabel(a4:$4$:\xoffca:\yoffca:0.75:90);
\arclabel(a3:$3$:\xoffca:\yoffca:0.75:270);

%

%
%
%
%
%
%
%
%
%

\end{tikzpicture}
        \caption{A CA graph and a normalized representation thereof. Every non-labeled edge corresponds to an $\overlap$-entry in the neighborhood matrix.}
        \label{fig:normrepr}  
    \end{figure}

    \begin{lemma}[Corollary 2.3.~\cite{hsu}]
        Every twin-free CA graph $G$ without a universal vertex has a normalized CA representation, that is $\mathcal{N}(G) \neq \emptyset$.    
    \end{lemma}  

    \begin{lemma}
        The canonical representation problem for CA graphs is logspace-reducible to the canonical representation problem for vertex-colored twin-free CA graphs without a universal vertex.  
    \end{lemma}
    \begin{proof}
        For a graph $G$ let $G_0$ denote the induced subgraph of $G$ that is obtained by removing all universal vertices from $G$ and only taking one vertex from each twin-class and deleting the rest. Let $c_0$ be a coloring of $G_0$ which assigns each vertex the cardinality of its twin class in $G$. It holds that $(G_0,c_0)$ and the number of universal vertices in $G$ suffice to reconstruct $G$.  
        Let $G$ be a CA graph. Compute the graph $(G_0,c_0)$. Since $(G_0,c_0)$ is twin-free and without universal vertices we can compute a canonical representation $\rho_0$ for it. For a vertex $v$ of $G$ let $v_0$ denote the twin of $v$ that occurs in $G_0$. A canonical representation of $G$ is given by $v \mapsto \rho_0(v_0)$ for every non-universal vertex $v$ of $G$ and every universal vertex of $G$ is represented by an arc which intersects with all other arcs. 
    \end{proof}
    
    Therefore for our purposes it suffices to consider only twin-free graphs without universal vertices and a vertex-coloring. 
    
    \begin{proviso}
        From this point on we assume every graph to be twin-free and without a universal vertex unless explicitly stated otherwise. As a consequence we view CA graphs as a set of CA matrices in the sense that the neighborhood matrix of every CA graph is a CA matrix.    
    \end{proviso}

\subparagraph*{Flips in Intersection Matrices.}
    \begin{table}[b]
    \caption{Algebraic flip functions $Z_{xy} \colon \itypeset \rightarrow \itypeset$}
    \label{tab:flip}
    \begin{center}
        \begin{tabular}{l | c c c c c}
            $Z_{xy}(\alpha)$ & \circlecover & \contained & \contains & \disjoint & \overlap \\
            \hline
            $Z_{00}$ & \circlecover & \contained & \contains & \disjoint & \overlap \\
            $Z_{01}$ & \contains & \disjoint & \circlecover & \contained & \overlap \\
            $Z_{10}$ & \contained & \circlecover & \disjoint & \contains & \overlap \\
            $Z_{11}$ & \disjoint & \contains & \contained & \circlecover & \overlap             
        \end{tabular}
    \end{center}
\end{table}   
McConnell \cite{mcc} observed that the operation of flipping arcs in CA models has a counterpart in intersection matrices. He called this counterpart operation algebraic flips. Note that for two arcs $A,B$ with intersection type $\alpha \in \itypeset$ the intersection type of $\overline{A}$ and $B$ is solely determined by $\alpha$. More precisely, the intersection type of $\overline{A}$ and $B$ is $Z_{10}(\alpha)$ where $Z_{10}$ is the function defined in Table~\ref{tab:flip}. Similarly, the intersection type of $A$ and $\overline{B}$ is given by $Z_{01}(\alpha)$. Using the functions $Z_{ij}$ we can define the operation of flipping a set of vertices in an intersection matrix.

\begin{definition}
	Let $\mu$ be an intersection matrix with vertex set $V$ and $X \subseteq V$. We define the intersection matrix $\mu^{(X)}$ obtained after flipping the vertices $X$ in $\mu$  as 
$$ \mu^{(X)}_{u,v} = Z_{ij}(\mu_{u,v}) \text{ with } i = 1 \text{ iff } u \in X \text{ and } j = 1 \text{ iff } v \in X     $$
    for all $u \neq v$ in $V$.
\end{definition}

Since flipping the same set of arcs twice is an involution it follows that $(\mu^{(X)})^{(X)} = \mu$. 

\begin{definition}
	Let $V$ be a set of vertices, let $\mathcal{A}$ be a set of arcs and let $\rho$ be a function that maps $V$ to $\mathcal{A}$. Then $\rho^{(X)} \colon V \rightarrow \mathcal{A}$ for $X \subseteq V$ is defined as follows:
    $$ \rho^{(X)}(v) = \begin{cases} 
    \overline{\rho(v)} & \text{, if $v \in X$} \\ 
    \rho(v) & \text{, if $v \notin X$} 
    \end{cases} $$          
\end{definition}    

Notice that flipping vertices in an intersection matrix is equivalent to flipping arcs in a CA representation in the following sense. 
Given an intersection matrix $\lambda$ and a subset of its vertices $X$ it holds that $\rho \in \mathcal{N}(\lambda) \Leftrightarrow \rho^{(X)} \in \mathcal{N}(\lambda^{(X)})$. Also, it is not difficult to observe that flipping is an isomorphism-invariant operation in the sense that flipping sets of vertices which are in the same orbit lead to isomorphic intersection matrices.

\section{Flip Trick}
In this section we generalize the idea used by Köbler, Kuhnert and Verbitsky in \cite{kob:hca} to compute canonical representations for Helly CA graphs. They showed that finding canonical representations for Helly CA graphs can be reduced to finding  canonical representations for vertex-colored interval matrices. We show that the idea behind this reduction also works for CA matrices in general. 
Recall that CA graphs can be seen as special case of CA matrices since the neighborhood matrix of every CA graph is a CA matrix. The converse does not hold, i.e.~there exist CA matrices which are not expressible as the neighborhood matrix of a CA graph (for instance any CA matrix with only two vertices that are not disjoint). 
The key result here, which is used in the subsequent sections, is that finding canonical representations for CA matrices is logspace-reducible to the task of computing what we call an invariant flip set function. 

McConnell showed in \cite{mcc} that CA representations for CA graphs can be computed as follows. Given a CA graph $G$ with neighborhood matrix $\lambda$ one can compute a set of vertices $X$ of $G$ such that $\lambda^{(X)}$ is an interval matrix. We call such a set $X$ a flip set.
Then by computing an interval representation $\rho$ for $\lambda^{(X)}$ and flipping back the arcs $X$ in $\rho$ one obtains a CA representation for $\lambda$ and therefore for $G$ as well \cite{mcc}. We essentially use the same argument to obtain canonical CA representations.

\begin{definition}
	Let $\lambda$ be a CA matrix. A subset of vertices $X$ of $\lambda$ is called a flip set if there exists a representation $\rho \in \mathcal{N}(\lambda)$ and a point $x$ on the circle such that $v \in X$ iff $\rho(v)$ contains the point $x$.
\end{definition}

The concept of flip sets has already been implicitly defined and used in both \cite{mcc} and \cite{kob:hca}. They observed that $\lambda^{(X)}$ is an interval matrix whenever $X$ is a flip set of a CA matrix $\lambda$. In fact, the other direction holds as well leading to the following characterization.

\begin{lemma}
	Let $\lambda$ be a CA matrix and $X$ is a subset of vertices of $\lambda$. It holds that $X$ is a flip set iff $\lambda^{(X)}$ is an interval matrix.    
    \label{lem:fsim}
\end{lemma}
\begin{proof}
    ``$\Rightarrow$'': 
    Let $X$ be a flip set of $\lambda$. Let $\rho \in \mathcal{N}(\lambda)$ be a witnessing representation of the fact that $X$ is a flip set, i.e.~there exists a point $x$ on the circle such that every arc $\rho(v)$ with $v \in X$ contains $x$ and every arc $\rho(v)$ with $v \notin X$ does not contain $x$. Consider the representation $\rho^{(X)} \in \mathcal{N}(\lambda^{(X)})$. It holds that no arc $\rho^{(X)}(v)$ with $v \in  V(\lambda)$ contains the point $x$ which implies that there is a hole in $\rho^{(X)}$ and thus $\lambda^{(X)}$ is an interval matrix.     
  
    ``$\Leftarrow$'': Let $X$ be a subset of vertices of $\lambda$ such that $\lambda^{(X)}$ is an interval matrix. We argue that $X$ must be a flip set. 
    Let $\rho \in \mathcal{N}(\lambda^{(X)})$ be a CA representation of $\lambda^{(X)}$ containing a hole at point $x$ on the circle. Such a representation must exist since $\lambda^{(X)}$ is an interval matrix. This means the arc $\rho(v)$ does not contain the point $x$ for every vertex $v \in V(\lambda)$. Consider the representation $\rho^{(X)} \in \mathcal{N}({(\lambda^{(X)})}^{(X)}) = \mathcal{N}(\lambda)$. 
    Then it can be checked that $\rho^{(X)}(v)$ contains the point $x$ iff $v$ is in $X$ and therefore $X$ is a flip set with respect to $\lambda$.   
\end{proof}

We already mentioned that the canonical representation problem for vertex-colored interval matrices can be solved in logspace due to \cite{kob:hca}. However, since the theorem that we reference just states this result for uncolored interval matrices we shortly explain how to modify the proof to incorporate the coloring, which is a straightforward task for anyone familiar with the proof. 

\begin{theorem}[{\cite[Thm.~5.5]{kob:hca}}]
    The canonical representation problem for vertex-colored interval matrices can be solved in logspace. 
    \label{thm:vcim_canon}
\end{theorem}
\begin{proof}
    In Theorem 5.5 of \cite{kob:hca} it is stated that a canonical interval representation for an interval matrix can be found in logspace. To prove this they convert the input interval matrix $\lambda$ into a colored tree $\mathbb{T}(\lambda)$ called $\Delta$ tree which is a complete invariant for interval matrices. The leafs of this tree correspond to the vertices of $\lambda$. By appending the color of a vertex from our vertex-colored interval matrix $\lambda$ to the existing color of its corresponding leave node in the colored $\Delta$ tree $\mathbb{T}(\lambda)$ one obtains a complete invariant for vertex-colored interval matrices. Then by applying the same argument given in the proof of Theorem 5.5 one can also compute a canonical representation for a vertex-colored interval matrix using this slightly modified colored $\Delta$ tree.
\end{proof}

A consequence of Lemma~\ref{lem:fsim} and Theorem~\ref{thm:vcim_canon} is that flip sets can be recognized in logspace. Given an intersection matrix $\lambda$ and a subset of vertices $X$ of $\lambda$ it suffices to check whether $\lambda^{(X)}$ is an interval matrix by trying to compute an interval representation.

\begin{definition}
     Let $\mathcal{C}$ be a class of CA matrices and $f$ is a vertex set selector.
     The function $f$ is called an invariant flip set function for $\mathcal{C}$ if the following conditions hold: 
     \begin{enumerate}
         \item For every $\lambda \in \mathcal{C}$ there exists an $X \in f(\lambda)$ such that $X$ is a flip set of $\lambda$
         \item $f$ is invariant for $\mathcal{C}$
     \end{enumerate}
     Recall that $f$ is globally invariant if $f$ is invariant for all intersection matrices.
     \label{def:cfsf}
\end{definition}

\begin{theorem}
	Let $\mathcal{C}$ be a class of CA matrices.
	The canonical representation problem for vertex-colored $\mathcal{C}$ is logspace-reducible to the problem of computing an invariant flip set function for $\mathcal{C}$.
    \label{thm:cfsf_cr}
\end{theorem}
\begin{proof}
    Let $f$ be an invariant flip set function for $\mathcal{C}$. 
    Given a vertex-colored CA matrix $(\lambda,c)$ with $ \lambda \in \mathcal{C}$ a canonical representation can be computed as follows.
    For every flip set $X \in f(\lambda)$ we associate it with the colored interval matrix $I_X = (\lambda^{(X)},c_X)$ where $c_X(v) = (c(v),\mathrm{red})$ if $v$ is in $X$ and $(c(v),\mathrm{blue})$ if $v$ is not in $X$ for all $v \in V(\lambda)$. For a colored interval matrix $I$ let $\hat{\rho}_I$ denote a canonical representation of $I$. Such a canonical representation can be computed in logspace due to Theorem~\ref{thm:vcim_canon}. Let $\hat{X}$ denote a flip set in $f(\lambda)$ such that the interval model of $\hat{\rho}_{I_{\hat{X}}}$ is lexicographically minimal, i.e.~for all flip sets $X$ in $f(\lambda)$ it holds that the model of $\hat{\rho}_{I_{X}}$ is not smaller than the model of $\hat{\rho}_{I_{\hat{X}}}$. Let $\hat{\rho}$ denote the CA representation that is obtained after flipping the red arcs in $\hat{\rho}_{I_{\hat{X}}}$. Since these are the arcs that were flipped to convert $\lambda$ into $I_X$ it holds that $\hat{\rho}$ is a representation of $\lambda$. To see that this leads to a canonical representation consider two isomorphic vertex-colored CA matrices $(\lambda,c)$ and $(\mu,d)$ with $\lambda,\mu \in \mathcal{C}$ and $V(\lambda)$ and $V(\mu)$ are disjoint. Let $\mathcal{I}_\lambda$ be the set of colored interval matrices $I_X$ for all flip sets $X \in f(\lambda)$, and the set $\mathcal{I}_\mu$ is defined analogously. 
    Let $\mathcal{M}_\lambda$ be the set of interval models $M$ such that there exists an $I \in \mathcal{I}_\lambda$ and $M$ is the model underlying the canonical representation $\hat{\rho}_I$ of $I$. The set $\mathcal{M}_\mu$ is  defined analogously. 
    Since $f$ is invariant it follows that for every $I \in \mathcal{I}_\lambda$ there exists an $I' \in \mathcal{I}_{\mu}$ such that $I$ and $I'$ are isomorphic, and vice versa. 
    Since the models in $\mathcal{M}_\lambda$ and  $\mathcal{M}_\mu$  only depend on the isomorphism type of the matrices in $\mathcal{I}_\lambda$ and $\mathcal{I}_\mu$ it follows that $\mathcal{M}_\lambda = \mathcal{M}_\mu$. The CA models which underlie the canonical representations of $\lambda$ and $\mu$ are both derived from the smallest element in $\mathcal{M}_\lambda = \mathcal{M}_\mu$ and thus are identical. 
\end{proof}

Suppose that there is a partition of the set of CA graphs into two classes $\mathcal{C}$ and $\mathcal{D}$ such that you can efficiently compute invariant flip set functions for both classes. One might be misled into thinking that this implies canonical representations for all CA graphs can be found efficiently. However, this is not the case unless the class $\mathcal{C}$ (or $\mathcal{D}$) can be efficiently  recognized, or one of the two invariant flip set functions is globally invariant. 

\begin{lemma}   
    Let $\mathcal{C}$ and $\mathcal{D}$ be classes of CA matrices. The canonical representation problem for $\mathcal{C} \cup \mathcal{D}$ is logspace-reducible to the canonical representation problem for $\mathcal{C}$ and the problem of computing a globally invariant flip set function for $\mathcal{D}$.
\label{lem:gir}    
\end{lemma}
\begin{proof}
  Let $f$ be a globally invariant flip set function for $\mathcal{D}$. Let $\mathcal{D}'$ be the set of CA matrices $\lambda$ such that $f(\lambda)$ contains a flip set. Clearly, $\mathcal{D}$ is a subset of $\mathcal{D}'$. It holds that $f(\lambda)$ contains a flip set iff $\lambda \in \mathcal{D}'$. Since $f$ is globally invariant it follows that $f$ is an invariant flip set function for $\mathcal{D}'$. 
  To obtain a canonical representation for a matrix $\lambda \in \mathcal{C} \cup \mathcal{D}$ first compute $f(\lambda)$. If $f(\lambda)$ contains a flip set it holds that $\lambda \in \mathcal{D}'$ and therefore the output of $f$ can be used to find a canonical representation for $\lambda$. If $f(\lambda)$ contains no flip set it must be the case that $\lambda \in \mathcal{C}$ and therefore the canonization algorithm for $\mathcal{C}$ can be applied.
\end{proof}

We conclude this section by restating the invariant flip set function that was used in \cite{kob:hca} to compute canonical representations for Helly CA graphs and explain why it is correct:

$$f_{\mathrm{HCA}}(G) = \big\{  N[u] \cap N[v]   \mid  u,v \in V(G) \big\}$$

In a Helly CA graph $G$ every inclusion-maximal clique $C$ of $G$ is a flip set. 
To see why this holds let $\rho$ be a representation of $G$ with the Helly property. Since $C$ is a clique this means every pair of arcs $\rho(u)$ and $\rho(v)$ with $u,v \in C$ intersects. By the Helly property it follows that the overall intersection $\bigcap_{v \in C} \rho(v)$ is non-empty. This means there exists a point $x$ on the circle such that every arc $\rho(v)$ with $v \in C$ contains $x$. Assume there exists a vertex $w \in V(G) \setminus C$ such that $\rho(w)$ contains $x$. This means $w$ must be adjacent to every vertex in $C$, which contradicts that $C$ is inclusion-maximal. Hence $C$ is a flip set. 
 
In \cite[Thm.~3.2]{kob:hca} it is shown that every Helly CA graph contains at least one inclusion-maximal clique which can be expressed as the common neighborhood of two vertices. 
Therefore $f_{\mathrm{HCA}}(G)$ returns at least one flip set for every Helly CA graph $G$. Also, it is trivial to see that $f_{\mathrm{HCA}}$ is globally invariant.

\section{Uniform Circular-Arc Graphs}
We define the class of uniform CA graphs for which computing a particular invariant flip set function reduces to computing a representation.
As a consequence, canonical representations for this class of CA graphs can be computed in polynomial-time. 
This is an interesting class for two reasons. First, it seems to capture the instances where it is easy to apply the flip trick. Secondly, its complement (within the CA graphs) is a rather exotic class of CA graphs with a quite particular structure. 
While the initial definition of uniformity makes it apparent why it suffices to find an arbitrary representation in order to obtain a canonical one, it is rather impractical when trying to understand what constitutes a uniform CA graph. 
We provide a more pleasant characterization of uniform CA graphs in terms of how certain triangles in a CA graph can be represented. 
This alternative characterization also reveals that every Helly CA graph is uniform. Additionally, we show that the canonical representation problem for uniform CA graphs is logspace-equivalent to what we call the non-Helly triangle representability problem. This problem is: given a CA graph $G$ and a set $T$ of three pairwise overlapping vertices as input, does there exist a representation $\rho$ of $G$ such that $T$ covers the whole circle in $\rho$?

The following kind of flip set will lead us to uniform CA graphs when trying to compute canonical representations. Given a CA matrix $\lambda$ recall that $X$ is a flip set of $\lambda$ if there exists a representation $\rho \in \mathcal{N}(\lambda)$ and a point $x$ on the circle such that $x \in \rho(v)$ iff $v \in X$ for all vertices $v$ of $\lambda$.
We impose the additional restriction that $x$ is not allowed to be an arbitrary point on the circle but instead has to be one of the endpoints in $\rho$. 

\begin{definition}
    Let $\lambda$ be a CA matrix and $u \in V(\lambda)$. A flip set $X$ of $\lambda$ is a $u$-flip set if there exists a representation $\rho \in \mathcal{N}(\lambda)$ and an endpoint $x$ of $\rho(u)$ such that $v \in X$ iff $\rho(v)$ contains the point $x$. 
\end{definition}

\begin{figure}[b]
    \centering
    \begin{minipage}{.5\textwidth}
        \centering
        \begin{tikzpicture}[shorten >=1pt,auto,node distance=1.2cm,
  main node/.style={circle,draw}]

    \centerarc[](0,0)(30:150:1.2);
    \centerarc[](0,0)(80:100:1.1);
    \centerarc[](0,0)(20:160:1.6);    
    \centerarc[](0,0)(30+20:150-360-20:0.9);    
    \centerarc[](0,0)(180:120:1.4);
    \centerarc[](0,0)(0:60:1.4);    

    \beam[](0:0:2:0.6:150)
    \beam[](0:0:2:0.6:30)
    \arclabel(xu1:$X_1$:0:0:2.2:150); 
    \arclabel(xu2:$X_2$:0:0:2.2:30); 
    
    \node (xu1) at ($({0.65*cos(110)},{0.65*sin(110)})$) {$u$};
    \beam[->]({0.65*cos(110)}:{0.65*sin(110)}:0.55:0.15:135)
        
    \node[circle,draw,inner sep=0,minimum size=0.1cm,fill=white] (l1) at ($({1.2*cos(30)},{1.2*sin(30)})$) {};
    \node[circle,draw,inner sep=0,minimum size=0.1cm,fill=white] (r1) at ($({1.2*cos(150)},{1.2*sin(150)})$) {};        

\end{tikzpicture}
        \captionof{figure}{Exemplary $u$-flip sets $X_1$ and $X_2$}
        \label{fig:uniform_flipset}
    \end{minipage}%
    \begin{minipage}{.5\textwidth}
        \centering
        \begin{tikzpicture}[shorten >=1pt,auto,node distance=1.2cm,
  main node/.style={circle,draw}]

\carc(0:0:90:120:1.3);
\arclabel(ul:$u$:0:0:1.45:90);

\carc(0:0:200:150:1.1);
\carc(0:0:350:190:0.9);
\arclabel(al:$a$:0:0:0.7:150);
\arclabel(bl:$b$:0:0:0.9:220);

\carc(0:0:50:120:1.1);
\carc(0:0:150:100:0.9);
\arclabel(cl:$c$:0:0:0.7:305);
\arclabel(dl:$d$:0:0:1.3:10);

\arclabel(up:$P_u=\{\{a,b\},\{c,d\}\}$:0:0:3:00);




\node[circle,draw,inner sep=0,minimum size=0.1cm,fill=white] (l1) at ($({1.3*cos(30)},{1.3*sin(30)})$) {};
\node[circle,draw,inner sep=0,minimum size=0.1cm,fill=white] (r1) at ($({1.3*cos(150)},{1.3*sin(150)})$) {};   

\end{tikzpicture}
        \captionof{figure}{Example of a $u$-overlap partition $P_u$}
        \label{fig:uovpart}
    \end{minipage}
\end{figure}

Clearly, every CA graph has a $u$-flip set for every vertex $u$. On the other hand, there are CA graphs that have flip sets which are not $u$-flip sets for any vertex $u$. For example, consider the cycle graph with $n \geq 4$ vertices. Every flip set that consists of exactly one vertex is not a $u$-flip set for any vertex $u$ of the cycle graph.

Consider the following task: given a CA graph $G$ and a vertex $u$, find a $u$-flip set of $G$. Clearly, no vertex $v$ which is disjoint from $u$ or contained by $u$ belongs to $X$ since in every representation the arc of $v$ does not contain any of the two endpoints of the arc of $u$. Similarly, if a vertex $v$ contains $u$ or forms a circle cover with $u$ then in every representation the arc of $v$ contains both endpoints of $u$ and therefore must be included in $X$. See Figure \ref{fig:uniform_flipset} for a schematic overview.

It remains to decide for the set of vertices $N^{\overlap}(u)$ that overlap with $u$ whether they should be included in $X$. A vertex $v$ which overlaps with $u$ contains exactly one of the endpoints of $u$ in any representation. 
Let $x,y$ be two vertices that overlap with $u$. We say $x$ and $y$ overlap from the same side with $u$ in $\rho$ if $\rho(x)$ and $\rho(y)$ contain the same endpoint of $\rho(u)$. Evidently, this is an equivalence relation with respect to $v$ and $\rho$ which partitions $N^{\overlap}(u)$ into two parts, namely the part which contains the left endpoint and the one which contains the right endpoint. If $X$ is a $u$-flip set then $X \cap N^{\overlap}(u)$ must be an equivalence class of the `overlap from the same side with $u$ in $\rho$'-relation for some $\rho \in \mathcal{N}(G)$.

\begin{definition}
    For a CA matrix $\lambda$ and a vertex $u$ of $\lambda$ we say a partition $Y$ of $N^{\overlap}(u)$ into two parts is a $u$-$\overlap$-partition if there exists a representation $\rho \in \mathcal{N}(\lambda)$ such that two vertices $x,y \in N^{\overlap}(u)$ are in the same part of $Y$ iff $\rho(x)$ and $\rho(y)$ overlap from the same side with $\rho(u)$.        
    We say $\overlap$-partition to mean an $u$-$\overlap$-partition for an arbitrary $u \in V(\lambda)$.
\end{definition}    

In general, for a vertex $u$ of a CA graph $G$ there can be multiple $u$-$\overlap$-partitions. In fact, there are instances with exponentially many $u$-$\overlap$-partitions with respect to $|N^{\overlap}(u)|$. A trivial way of obtaining at least one $u$-$\overlap$-partition for every vertex $u$ of a CA graph $G$ is to compute an arbitrary representation $\rho \in \mathcal{N}(G)$. But
the $\overlap$-partitions obtained by this method are not invariant and thus do not yield canonical representations. However, if one considers CA graphs where there is only one $u$-$\overlap$-partition for every vertex $u$ then an arbitrary representation suffices.

\begin{definition}[Uniform CA Graphs]
    A CA graph $G$ is uniform if for every vertex $u$ in $G$ there exists exactly one $u$-$\overlap$-partition. This partition is denoted by $P_u = \{P_{u,1}, P_{u,2}\}$.
    \label{fact:uovpart}
\end{definition}

\begin{lemma}
    The following mapping is an invariant flip set function for uniform CA graphs. Let $G$ be a uniform CA graph.
    $$ \Funiform(G) = \bigcup_{\substack{u \in V(G) \\ i \in \{1,2\}}}  
    \big\{ \{ u \} \cup N^{\contained}(u) \cup N^{\circlecover}(u) \cup P_{u,i} \big\} 
    $$   
    \label{lem:funiform}
\end{lemma}
\begin{proof}
	Let $G$ be a uniform CA graph and $X$ is in $\Funiform(G)$ with $X =  \{ u \} \cup N^{\contained}(u) \cup N^{\circlecover}(u) \cup P_{u,i}$ for some $u \in V(G)$ and $i \in \{1,2\}$. It follows from Figure \ref{fig:uniform_flipset} and the definition of $\overlap$-partitions that $X$ is a $u$-flip set. The invariance of $\Funiform(G)$ follows from the fact that the	 intersection type of two vertices as well as the property of being an $\overlap$-partition is independent of the vertex labels.  
\end{proof}

We remark that the function $\Funiform$ is undefined for non-uniform CA graphs since the sets $P_{u,1}$ and $P_{u,2}$ are not well-defined in that context.

\begin{theorem}
    Canonical representations for uniform CA graphs can be computed in polynomial-time.
    \label{thm:uca_canon}
\end{theorem}
\begin{proof}
	Let $G$ be a uniform CA graph. Compute a normalized representation $\rho$ of $G$ and extract the $u$-$\overlap$-partition for each vertex $u$ from $\rho$. Then compute $\Funiform(G)$ from Lemma~\ref{lem:funiform} to obtain a canonical CA representation for $G$. Since CA representations can be computed in polynomial-time (see for instance \cite{mcc}) it follows that this procedure also works in polynomial-time.
\end{proof}

Considering that our definition of uniform CA graphs arose from the desire to compute invariant $u$-flip sets one might expect that these graphs are only a small special case of CA graphs.
Surprisingly, quite the opposite is the case as we will see. We give an alternative definition of uniform CA graphs which gives a better intuition as to why many CA graphs are uniform.

\begin{definition}
	Let $\lambda$ be a CA matrix.
	An $\overlap$-triangle $T$ of $\lambda$ is a set of three vertices that overlap pairwise, i.e.~for all $u \neq v$ in $T$ it holds that $u \: \overlap \: v$.    
	An $\overlap$-triangle $T$ is representable as non-Helly triangle (interval triangle) if there exists a representation $\rho \in \mathcal{N}(\lambda)$ such that the set of arcs $\set{\rho(x)}{x \in T}$ does (not) cover the whole circle.  Let $\ovtnht(\lambda)$ and $\ovtit(\lambda)$ denote the sets of $\overlap$-triangles representable as non-Helly triangles  and interval triangles respectively.  
    \label{def:ovtriangle}   
\end{definition}

This definition also applies to CA graphs via their neighborhood matrix, i.e.~$\ovtit(G) = \ovtit(\lambda)$ and $\ovtnht(G) = \ovtnht(\lambda)$ where $\lambda$ is the neighborhood matrix of $G$.  
See Figure \ref{fig:ucathm} for an example where the vertices $u,x,z$ are represented as non-Helly triangle on the left and interval triangle on the right.

Recall that a set of arcs which intersect pairwise but have overall empty intersection is called non-Helly. Since three pairwise overlapping arcs that cover the whole circle have overall empty intersection we call such a set a non-Helly triangle. In fact, one can verify that this is the only non-Helly arrangement of three arcs. A complete list of inclusion-minimal non-Helly CA models can be found in \cite[Corrollary~3.1]{joe}. 

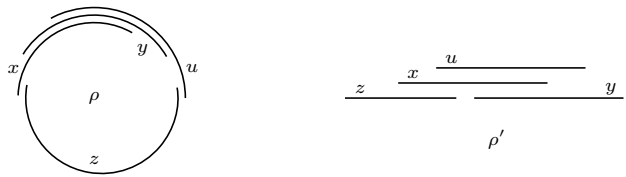
\begin{figure}[b]
    \centering
        \begin{tikzpicture}[shorten >=1pt,auto,node distance=1.2cm,
  main node/.style={circle,draw}]

\newcommand*{\xoffca}{0}%
\newcommand*{\yoffca}{0}%

\newcommand*{\xoff}{4}%
\newcommand*{\yoff}{0}%

\newcommand*{\movea}{1}%

\carc(\xoffca:\yoffca:60:120:1.2);
\carc(\xoffca:\yoffca:90:120:1.1);
\carc(\xoffca:\yoffca:120:120:1.0);
\carc(\xoffca+0.1:\yoffca:270:200:1.0);

\arclabel(rhol:$\rho$:\xoffca:\yoffca:0:0);

\arclabel(vz1:$z$:\xoffca:\yoffca:0.82:270);%
\arclabel(vy1:$y$:\xoffca:\yoffca:0.9:45);%
\arclabel(vu1:$u$:\xoffca:\yoffca:1.35:17);%
\arclabel(vx1:$x$:\xoffca:\yoffca:1.13:160);%

\draw[line width=0.6pt] (\xoff+0.5,\yoff+0.4) -- (\xoff+2+0.5,\yoff+0.4);
\draw[line width=0.6pt] (\xoff,\yoff+0.2) -- (\xoff+2,\yoff+0.2);
\draw[line width=0.6pt] (\xoff+2*0.5,\yoff) -- (\xoff+2+2*0.5,\yoff);
\draw[line width=0.6pt] (\xoff-0.7,\yoff) -- (\xoff+1.5-0.7,\yoff);

\node (vz) at (\xoff-0.7+0.2,\yoff+0.12) {$z$};
\node (vx) at (\xoff+0.2,\yoff+0.2+0.12) {$x$};
\node (vu) at (\xoff+0.5+0.2,\yoff+0.4+0.12) {$u$};
\node (vy) at (\xoff+2+2*0.5-0.2,\yoff+0.12) {$y$};
\node (vrho) at (\xoff+0.3+1,\yoff-0.55) {$\rho'$};

\end{tikzpicture}
    \caption{``$\Leftarrow$''-direction in the proof of Theorem \ref{thm:uca_char}}
    \label{fig:ucathm}                
\end{figure}  

\begin{theorem}
    A CA graph $G$ is uniform iff $\ovtit(G) \cap \ovtnht(G) = \emptyset$. 
    \label{thm:uca_char}
\end{theorem}
\begin{proof}
    ``$\Rightarrow$'': Assume there exists a uniform CA graph $G$ with $\ovtit(G) \cap \ovtnht(G) \neq \emptyset$. Let $T$ be an $\overlap$-triangle in $\ovtit(G) \cap \ovtnht(G)$ and $T = \{x,y,z\}$. This means there exist two representations $\rho_I,\rho_N \in \mathcal{N}(G)$ such that $T$ is represented as interval triangle in $\rho_I$ and as non-Helly triangle in $\rho_N$. We assume w.l.o.g.~that $\rho_I(y) \subset \rho_I(x) \cup \rho_I(z)$, i.e.~$y$ is placed in-between $x$ and $z$ in $\rho_I$. This means $y$ and $z$ must be in the same part of the unique $x$-$\overlap$-partition $P_x$. However, $y$ and $z$ do not contain the same endpoint of $x$ in the representation $\rho_N$, which contradicts that $G$ is uniform. 
    
    ``$\Leftarrow$'': Assume there exists a CA graph $G$ with $\ovtit(G) \cap \ovtnht(G) = \emptyset$ that is not uniform. This means there exist a vertex $u$, two vertices $x,y \in N^\overlap(u)$ and two representations $\rho,\rho' \in \mathcal{N}(G)$ such that $x$ and $y$ overlap from the same side with $u$ in $\rho$ but not in $\rho'$. This implies that $x$ and $y$ must overlap and therefore $T=\{u,x,y\}$ is an $\overlap$-triangle. Notice that $T$ must be represented as interval triangle in $\rho$ because $x$ and $y$ both contain the same endpoint of $u$.
    It holds that $T$ is represented as interval triangle in $\rho'$ as well since otherwise $T \in \ovtit(G) \cap \ovtnht(G)$. Also, we assume w.l.o.g.~that $\rho(y) \subset \rho(x) \cup \rho(u)$. Since $u$ and $y$ overlap it holds that $N[u] \setminus N[y] \neq \emptyset$. Due to $\rho'$ it follows that $N[u] \setminus N[y] \subseteq N[u] \cap N[x]$. For a vertex $z \in N[u] \setminus N[y]$ to intersect with both $u$ and $x$ it is necessary that $z$ overlaps with $u$ and $x$ due to the representation $\rho$. It follows that $\{u,x,z\}$ is represented as non-Helly triangle in $\rho$. On the other hand, $\{u,x,z\}$ must be represented as interval triangle in $\rho'$ and therefore  $\ovtit(G) \cap \ovtnht(G) \neq \emptyset$, contradiction. See Figure \ref{fig:ucathm} for a schematic overview of $\rho$ and $\rho'$. 
\end{proof}

Observe that if an $\overlap$-triangle $T$ of $G$ is representable as non-Helly triangle then this implies that $T$ must have certain structural properties in $G$. For example, every vertex of $G$ must be adjacent to at least one of the vertices in $T$ since $T$ covers the whole circle in some representation. Similarly, if $T$ is representable as interval triangle this also implies some structural properties.  For instance, there must be an $x \in T$ such that every vertex that is adjacent to $x$ must also be adjacent to at least one other vertex in $T$. If an $\overlap$-triangle is representable as both non-Helly triangle and interval triangle then it must satisfy all of these structural properties at once. As a consequence such an $\overlap$-triangle must have a very particular structure which extends to the whole graph as we will see in the next section.  

A CA graph is Helly if it has a Helly CA representation. In \cite[Theorem~4.1]{joe} it is shown that every `stable' representation of a Helly CA graph is Helly. Since every normalized representation has the `stable' property it follows that a CA graph is Helly iff every normalized representation of it is Helly. If a CA graph $G$ is Helly this implies that $\ovtnht(G)$ is empty, and therefore every Helly CA graph is uniform.

A natural question to consider is the computational complexity of deciding whether an  $\overlap$-triangle is representable as non-Helly triangle or interval triangle. Given a CA graph $G$ and an $\overlap$-triangle $T$ of $G$ let us call the problem of deciding whether $T$ is in $\ovtnht(G)$ the non-Helly triangle representability problem. Analogously, deciding whether $T$ is in $\ovtit(G)$ is called the interval triangle representability problem. 
In the case of uniform CA graphs these two problems are complementary, i.e.~an $\overlap$-triangle $T$ is in $\ovtnht(G)$ iff $T$ is not in $\ovtit(G)$. In the following, we show that solving either of these two problems for uniform CA graphs is logspace-equivalent to computing a canonical representation for uniform CA graphs. 
 
\begin{definition}
    Let $G$ be a CA graph and $T=\{u,v,w\}$ is an $\overlap$-triangle of $G$. We say $v$ is amidst $u$ and $w$ if one of the following conditions holds:
    \begin{enumerate}
        \item $N_T(u)$ and $N_T(w)$ are non-empty
        \item there exists a $z \in N_T(u,w)$ such that $\{u,w,z\} \in \ovtnht(G)$
    \end{enumerate}
    \label{def:amidst}
\end{definition}

\begin{lemma}
    \label{lem:inbtw}
    Let $G$ be a uniform CA graph and $T=\{u,v,w\}$ is an $\overlap$-triangle of $G$ with $T \notin \ovtnht(G)$. Then the following statements are equivalent:
    \begin{enumerate}
        \item $v$ is amidst $u$ and $w$
        \item $\exists \rho \in \mathcal{N}(G): \rho(v) \subset \rho(u) \cup \rho(w) $
        \item $\forall \rho \in \mathcal{N}(G): \rho(v) \subset \rho(u) \cup \rho(w) $
    \end{enumerate}  
    \label{lem:amidst}              
\end{lemma}
\begin{proof}
    ``2 $\Rightarrow$ 1'': Let $\rho$ be in $\mathcal{N}(G)$ such that $\rho(v) \subset \rho(u) \cup \rho(w)$ and assume that $v$ is not amidst $u,w$. Since $v$ overlaps with $u$ and $w$ it holds that $N[u] \setminus N[v]$ and $N[w] \setminus N[v]$ are non-empty. Because $N_T(u) = N_T(w) = \emptyset$ it must hold that $N_T(u,w) \neq \emptyset$. Let $z \in N_T(u,w)$. For $z$ to intersect with $u$ and $w$ in $\rho$ it must hold that $\{u,w,z\}$ is represented as non-Helly triangle in $\rho$. This contradicts the assumption that $v$ is not amidst $u,w$.  
    
    ``1 $\Rightarrow$ 3'': Let $v$ be amidst $u$ and $w$ and assume that there exists a $\rho \in \mathcal{N}(G)$ such that $\rho(v) \not\subset \rho(u) \cup \rho(w)$. Since $T \notin \ovtnht(G)$ and $G$ is uniform it follows by Theorem \ref{thm:uca_char} that $T$ must be represented as interval triangle in every representation, which includes $\rho$. We assume w.l.o.g.~that $\rho(w) \subset \rho(u) \cup \rho(v)$. From that it follows that $N_T(w)$ is empty and therefore there must be a $z \in  N_T(u,w)$ such that $\{u,w,z\}$ is a non-Helly triangle in $\rho$, which is impossible. 
    
    ``3 $\Rightarrow$ 2'': clear.
\end{proof}

\begin{definition}
	Let $G$ be a CA graph and $u \in V(G)$. Let the binary relation $\sim_u$ on $N^{\overlap}(u)$ be defined such that $x \sim_u y$ holds if one of the following holds:
	\begin{enumerate}
		\item $x = y$
		\item $x \: \contained \: y$ or $x \: \contains \: y$
		\item $x \: \overlap \: y$, $\{u,x,y\} \notin \ovtnht(G)$ and $u$ is not amidst $x$ and $y$
	\end{enumerate} 
	\label{def:simu}
\end{definition}

\begin{lemma}
    For every uniform CA graph $G$ and $u \in V(G)$ it holds that the partition induced by $\sim_u$ equals the unique $u$-$\overlap$-partition $P_u$. Stated differently, $x \sim_u y$ iff $x$ and $y$ are in the same part of $P_u$.
    \label{lem:simu}
\end{lemma}
\begin{proof}
    ``$\Rightarrow$'': Let $x \sim_u y$ and assume for the sake of contradiction that $x$ and $y$ are not in the same part of the $u$-$\overlap$-partition. This means there exists a representation $\rho \in \mathcal{N}(G)$ such that $\rho(x)$ and $\rho(y)$ contain different endpoints of $\rho(u)$. This is only possible if $x$ and $y$ overlap. Since $\{u,x,y\} \notin \ovtnht(G)$ this means $\{u,x,y\}$ must be represented as interval triangle in $\rho$. In order for $\rho(x)$ and $\rho(y)$ to contain different endpoints of $\rho(u)$ it must hold that $\rho(u) \subset \rho(x) \cup \rho(y)$, which implies that $u$ is amidst $x$ and $y$ by Lemma \ref{lem:amidst}. This contradicts $x \sim_u y$.  
    
    ``$\Leftarrow$'': Let $x$ and $y$ be in the same part of the $u$-$\overlap$-partition and assume that $x \sim_u y$ does not hold. This implies that $x$ and $y$ must overlap and therefore $\{u,x,y\}$ form an $\overlap$-triangle. For $x \sim_u y$ to not hold it must be either the case that $\{u,x,y\}$ is only representable as non-Helly triangle or $u$ is amidst $x$ and $y$. In both cases this contradicts $x$ and $y$ being in the same part of the $u$-$\overlap$-partition.
\end{proof}

\begin{theorem}
    The representation, canonical representation, non-Helly triangle representability and interval triangle representability problem for uniform CA graphs are logspace-equivalent.
    \label{thm:eq_problem}
\end{theorem}
\begin{proof}
    The non-Helly triangle representability and interval triangle representability problem for uniform CA graphs are logspace-equivalent because they are complementary in the sense that an $\overlap$-triangle is representable as non-Helly triangle iff it is not representable as interval triangle. This follows from the fact that an $\overlap$-triangle can only be either represented as non-Helly triangle or interval triangle and these two possibilities are mutually exclusive in the case of uniform CA graphs. As a consequence these two problems are trivially reducible to the representation problem for uniform CA graphs. Given a uniform CA graph $G$, an $\overlap$-triangle $T$ of $G$ and a representation $\rho \in \mathcal{N}(G)$ it holds that $T \in \ovtnht(G)$ iff $T \notin \ovtit(G)$ iff $T$ is represented as non-Helly triangle in $\rho$. 
    
    The representation problem is obviously reducible to the canonical representation problem. Therefore it remains to show that the canonical representation problem for uniform CA graphs is reducible to the non-Helly triangle representability problem. To obtain a canonical representation for a uniform CA graph we can use the invariant flip set function given in Lemma \ref{lem:funiform}. To compute this function we need to figure out the unique $\overlap$-partitions for each vertex. By Lemma \ref{lem:simu} this can be done by computing the equivalence relation $\sim_u$ for each vertex $u$. It can be verified that this relation is computable in logspace using queries of the form $T \in \ovtnht(G)$.
\end{proof}

The isomorphism problem for CA graphs can be reduced to the one for non-uniform CA graphs in polynomial-time due to Theorem~\ref{thm:uca_canon}. However,  a reduction from the canonical representation problem for CA graphs to the one for non-uniform CA graphs does not immediately follow from Theorem~\ref{thm:uca_canon} unless uniform CA graphs can be recognized in polynomial-time. An alternative approach to construct such a reduction is to solve the non-Helly triangle representability problem for uniform CA graphs with an additional requirement.

\begin{definition}
	The globally invariant non-Helly triangle representability problem for uniform CA graphs is defined as follows. Let $A$ be an algorithm that correctly decides the non-Helly triangle representability problem for uniform CA graphs. Let $f_A$ be the function computed by $A$, i.e.~for a graph $G$ and an $\overlap$-triangle $T$ of $G$ it holds that $f_A(G,T) = 1$ iff $A$ accepts $(G,T)$. We say $A$ decides the globally invariant non-Helly triangle representability problem for uniform CA graphs if $f_A$ is an invariant \emph{for all graphs}. Stated differently, the output of $A$ must be independent of the vertex labels.	     
\end{definition} 

\begin{lemma}
	The canonical representation problem for CA graphs is logspace-reducible to the 
	globally invariant non-Helly triangle representability problem for uniform CA graphs and the canonical representation problem for vertex-colored non-uniform CA graphs.	
    \label{lem:ginhtp}
\end{lemma}
\begin{proof}
	Suppose we are given an algorithm $A$ which solves the 	globally invariant non-Helly triangle representability problem for uniform CA graphs. We argue that $A$ can be used to compute a globally invariant flip set function for uniform CA graphs. 
	From Lemma \ref{lem:gir} it then follows that the canonical representation problem for CA graphs reduces to that for vertex-colored non-uniform CA graphs. 
	
    Given a CA graph $G$ let $\Delta(G,A)$ be the set of $\overlap$-triangles $T$ of $G$ such that $A$ accepts $(G,T)$. If $G$ is a uniform CA graph then $\Delta(G,A) = \ovtnht(G)$. Consider Definition \ref{def:amidst} and \ref{def:simu} and suppose that each occurrence of $\ovtnht(G)$ is replaced by $\Delta(G,A)$. Let us call the new relation $\sim_u^A$. Clearly, in the case of uniform CA graphs $\sim_u$ and $\sim_u^A$ coincide.  Next, consider the following variant of $\Funiform$:
    $$ \Funiform^A(G) = \bigcup_{\substack{u \in V(G) \\ X \in (N^{\overlap}(u) / \sim_u^A)}}  
    \big\{ \{ u \} \cup N^{\contained}(u) \cup N^{\circlecover}(u) \cup X \big\} 
    $$   
    where $(N^{\overlap}(u) / \sim_u^A)$ denotes the equivalence classes of $\sim_u^A$. If $\sim_u^A$ is not an equivalence relation let $(N^{\overlap}(u) / \sim_u^A) = \emptyset$. If $G$ is a uniform CA graph then it follows from Lemma~\ref{lem:simu} that $\Funiform(G) = \Funiform^A(G)$. Therefore $\Funiform^A$ is an invariant flip set function for uniform CA graphs. Additionally, it can be verified that $\Funiform^A$ is globally invariant due to the fact that the answer of $A$ is independent of the vertex labels. Also, the function $\Funiform^A$ can be computed in logspace using queries of the form $T \in \Delta(G,A)$. Observe that $\Delta(G,A)$ only provides $n^3$ bits of information with $n = |V(G)|$ and therefore can be computed `in a single query' by a functional oracle which outputs the $n^3$ bits of information. 
\end{proof}

\section{Non-Uniform CA Graphs and Restricted CA Matrices}
In the first part of this section we examine the structure of non-uniform CA graphs. Every such graph must have two $\overlap$-triangles which have exactly one vertex in common and both are representable as interval triangle and as non-Helly triangle. This pair of $\overlap$-triangles enforces a particular structure in non-uniform CA graphs. In the second part we introduce restricted CA matrices, which try to partly capture this structure. Roughly speaking, restricted CA matrices can be seen as a generalization of the neighborhood matrices of non-uniform CA graphs. We pay the price of considering this more general class of structures in order to provide a logspace reduction from the canonical representation problem for CA graphs to that of restricted CA matrices. 

\begin{definition}
    Given a CA graph $G$, an induced 4-cycle $C = (u,w,w',u')$ of $G$ and $v \in V(G) \setminus C$. We say $(C,v)$ is a non-uniformity witness of $G$ if $\{u,v,w \}, \{ u',v,w' \} \in \ovtit(G) \cap \ovtnht(G)$. We also simply call $(C,v)$ a witness of $G$. 
\end{definition}

\begin{theorem}
    A CA graph $G$ is non-uniform iff $G$ has a non-uniformity witness.
    \label{thm:nuwc}
\end{theorem}
\begin{proof}
    ``$\Rightarrow$'': Let $G$ be a non-uniform CA graph. Due to Theorem \ref{thm:uca_char} there exists an $\overlap$-triangle $T$ of $G$ with $T \in \ovtit(G) \cap \ovtnht(G)$. Let $T= \{u,v,w\}$ and $\rho_{\text{I}} \in \mathcal{N}(G)$ such that $v$ is in-between $u$ and $w$, i.e.~$\rho_{\text{I}}(v) \subset \rho_{\text{I}}(u) \cup \rho_{\text{I}}(w)$. First, we show that there exists an induced 4-cycle $C = (u,w,w',u')$ in $G$. 
       
    From the non-Helly triangle representation of $T$ it follows that $N[u] \cup N[v] \cup N[w] = V(G)$. Since $v$ is in-between $u$ and $w$ this means $N[u] \cup N[w] = V(G)$. It holds that $u$ and $w$ overlap. Therefore one of the conditions in the definition of the neighborhood matrix for $u$ and $w$ to form a circle cover must be violated. Let us assume w.l.o.g.~that the violated condition is that there exists a $u' \in N[u] \setminus N[w]$ such that $N[u'] \not\subseteq N[u]$. This means $u'$ must overlap with $u$ and there exists a $w' \in N[u'] \setminus N[u]$. Since $w' \notin N[u]$ it follows from   $N[u] \cup N[w] = V(G)$ that $w' \in N[w]$ and because $w$ is disjoint from $u'$, and because $w'$ intersects with both $u'$ and $w$ it follows that $w'$ overlaps with $u'$ and $w$. Therefore $C = (u,w,w',u')$ is an induced 4-cycle in $G$. 
    
    It remains to show that $\{u',v,w'\}$ is an $\overlap$-triangle and that it is in both $\ovtit(G)$ and $\ovtnht(G)$. Consider the representation $\rho_{\text{I}}$ from before. Assume for the sake of contradiction that $v$ does not overlap with $u'$. Then due to $\rho_{\text{I}}$ it must be the case that $u'$ is disjoint from $v$ and thus $u' \in N_T(u)$. However, due to fact that $T$ is representable as non-Helly triangle this would imply that $u'$ is contained by $u$, which is not the case. Therefore $u'$ overlaps with $v$ as the other intersections types are out of question. For the same reason $w'$ overlaps with $v$ and hence $T'= \{u',v,w'\}$ is an $\overlap$-triangle. Now, it can be verified that in every representation of $G$ where $T$ is a non-Helly triangle it follows that $T'$ must be an interval triangle and vice versa. This concludes that $T'$ is in $\ovtit(G) \cap \ovtnht(G)$.
    
    ``$\Leftarrow$'': Follows directly from Theorem \ref{thm:uca_char}.
\end{proof}

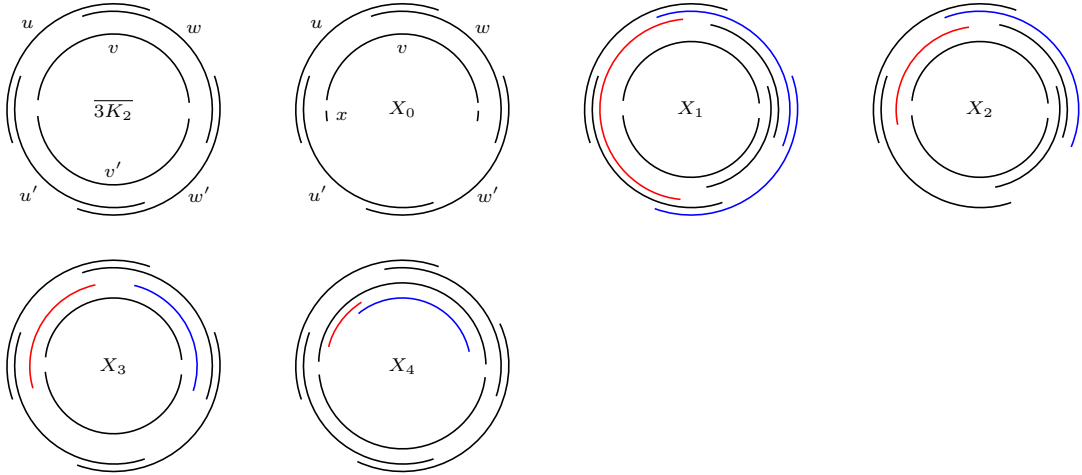
\begin{figure}
	\centering
	\begin{tikzpicture}[shorten >=1pt,auto,node distance=1.2cm,
  main node/.style={circle,draw}]

\newcommand*{\xoff}{0}%
\newcommand*{\xoffa}{3.8}%
\newcommand*{\xoffb}{2*3.8}%
\newcommand*{\xoffc}{3*3.8}%

\newcommand*{\yoff}{0}%
\newcommand*{\yoffa}{-3.4}%


\arclabel(m1:$u$:\xoff:\yoff:1.6:135);%
\carc(\xoff:\yoff:135:130:1.4);
\arclabel(m1:$w'$:\xoff:\yoff:1.6:315);%
\carc(\xoff:\yoff:315:130:1.4);

\arclabel(m1:$w$:\xoff:\yoff:1.5:45);%
\carc(\xoff:\yoff:45:130:1.3);

\arclabel(m1:$u'$:\xoff:\yoff:1.57:225);%
\carc(\xoff:\yoff:225:130:1.3);

\arclabel(m1:$v$:\xoff:\yoff:0.8:90);%
\carc(\xoff:\yoff:90:170:1);

\arclabel(m1:$v'$:\xoff:\yoff:0.8:270);%
\carc(\xoff:\yoff:270:170:1);
\arclabel(m1:$\overline{3K_2}$:\xoff:\yoff:0:0);%

\arclabel(m1:$u$:\xoffa:\yoff:1.6:135);%
\arclabel(m1:$w'$:\xoffa:\yoff:1.6:315);%
\arclabel(m1:$w$:\xoffa:\yoff:1.5:45);%
\arclabel(m1:$u'$:\xoffa:\yoff:1.57:225);%
\arclabel(m1:$v$:\xoffa:\yoff:0.8:90);%
\carc(\xoffa:\yoff:135:130:1.4);
\carc(\xoffa:\yoff:315:130:1.4);
\carc(\xoffa:\yoff:45:130:1.3);
\carc(\xoffa:\yoff:225:130:1.3);
\carc(\xoffa:\yoff:90:170:1);
\carc(\xoffa:\yoff:186:10:1);
\arclabel(m1:$x$:\xoffa:\yoff:0.8:186);%
\carc(\xoffa:\yoff:186+170:10:1);
\arclabel(m2:$X_0$:\xoffa:\yoff:0:0);%

\carc(\xoffb:\yoff:135:130:1.4);  
\carc(\xoffb:\yoff:225:130:1.3);  
\carc(\xoffb:\yoff:45-15:98:1.15);  
\carc(\xoffb:\yoff:315+15:98:1.05); 
\carc(\xoffb:\yoff:90:170:0.9);
\carc(\xoffb:\yoff:270:170:0.9);

\carcblue(\xoffb:\yoff:315:130:1.4);  
\carcblue(\xoffb:\yoff:45:134:1.3);   
\carcred(\xoffb:\yoff:180:170:1.2);

\arclabel(m3:$X_1$:\xoffb:\yoff:0:0);%

\carc(\xoffc:\yoff:135:130:1.4);  
\carc(\xoffc:\yoff:225:130:1.3);  
\carc(\xoffc:\yoff:45-15:98:1.15);  
\carc(\xoffc:\yoff:315+15:98:1.05); 
\carc(\xoffc:\yoff:90:170:0.9);
\carc(\xoffc:\yoff:270:170:0.9);

\carcblue(\xoffc:\yoff:45:135:1.3);   
\carcred(\xoffc:\yoff:145:95:1.1);

\arclabel(m3:$X_2$:\xoffc:\yoff:0:0);%

\carc(\xoff:\yoffa:135:130:1.4);
\carc(\xoff:\yoffa:315:130:1.4);
\carc(\xoff:\yoffa:45:130:1.3);
\carc(\xoff:\yoffa:225:130:1.3);
\carc(\xoff:\yoffa:90:170:0.9);
\carc(\xoff:\yoffa:270:170:0.9);

\carcred(\xoff:\yoffa:150:95:1.1);
\carcblue(\xoff:\yoffa:30:95:1.1);

\arclabel(m1:$X_3$:\xoff:\yoffa:0:0);%

%
%

\carc(\xoffa:\yoffa:135:130:1.4);
\carc(\xoffa:\yoffa:315:140:1.4);
\carc(\xoffa:\yoffa:42:120:1.3);
\carc(\xoffa:\yoffa:225:130:1.3);
\carc(\xoffa:\yoffa:90:178:1.1);
\carc(\xoffa:\yoffa:270:168:1.1);

\carcred(\xoffa:\yoffa:145:45:1.0);
\carcblue(\xoffa:\yoffa:72:120:0.9);

\arclabel(m1:$X_4$:\xoffa:\yoffa:0:0);%

%
%
%
%

\end{tikzpicture}
	\caption{Examples of non-uniform CA graphs and one uniform CA graph $X_4$}
	\label{fig:ex_nuca}                
\end{figure}  

In Figure~\ref{fig:ex_nuca} five non-uniform CA graphs and one uniform CA graph ($X_4$) are given by their CA models. We explain how to verify this claim. First, we have to check that every CA model is normalized. This means the graphs which are induced by these models must be twin-free and without a universal vertex. Additionally, the intersection types of the arcs must match the intersection types in the induced graph (or more precisely its neighborhood matrix). A quick way to determine whether two overlapping arcs also overlap in the graph is to check if they jointly occur in an induced $n$-cycle for some $n \geq 4$.  

To see that the first five CA graphs are non-uniform we have to find an $\overlap$-triangle that is representable as both interval and non-Helly triangle. In the case of $\overline{3K_2}$ this $\overlap$-triangle can be $\{u,v,w\}$. In the given representation $\{u,v,w\}$ is represented as interval triangle. Observe that $v$ and $v'$ are in the same orbit and therefore the labels $v$ and $v'$ can be swapped in the representation. After swapping $v$ and $v'$ the $\overlap$-triangle $\{u,v,w\}$ is represented as non-Helly triangle. For the graph $X_0$ we can also choose the $\overlap$-triangle $\{u,v,w\}$. In this case there is an automorphism which swaps $u$ with $u'$ and $w$ with $w'$ and has the other vertices as fix-points. After changing the labels in the representation according to this automorphism it holds that $\{u,v,w\}$ is represented as non-Helly triangle. We remark that $\overline{3K_2}$ and $X_0$ are minimal in the sense that no induced subgraph of them is a non-uniform CA graph. Next, let us consider the graphs $X_1$ to $X_3$. Observe that the black arcs in each of these graphs form an induced $\overline{3K_2}$ subgraph. We assume that the black arcs are labeled with $u,u',v,v',w,w'$ in the same way that the representation of $\overline{3K_2}$ is labeled. It holds that $v$ and $v'$ are in the same orbit in all of these four graphs because they have the same open neighborhood. Therefore  $\{u,v,w\}$ is representable as both interval and non-Helly triangle due to the same argument that we made for $\overline{3K_2}$.

To show that $X_4$ is uniform we argue that it has a unique normalized representation, i.e.~$|\mathcal{N}(X_4)| = 1$.
Observe that this graph has a unique CA model. Additionally, it has no non-trivial automorphism (it is rigid). Therefore $X_4$ has a unique CA representation.

\begin{fact}
	Every non-uniform CA graph contains $\overline{3K_2}$ or $X_0$ as induced subgraph.
\end{fact}
\begin{proof}
    Let $G$ be a non-uniform CA graph. Due to Theorem \ref{thm:nuwc} there exists a witness $(C,v)$ of $G$ with $C=(u,w,w',u')$.
    Since $G$ does not contain a universal vertex it holds that $V(G) \setminus N[v]$ is non-empty. Due to the fact that $\{u,v,w \}$ and $\{u',v,w' \}$ can be represented as interval triangles it follows that $N_C(C \setminus \{x\}) \subseteq N[v]$ for all $x \in C$. Therefore $V(G) \setminus N[v] \subseteq N_C(C) \cup N_C(u,u') \cup N_C(w,w')$. Suppose there is a $v' \in N_C(C) \setminus N[v]$. Then the vertices of $C$ along with $v$ and $v'$ form an induced $\overline{3K_2}$-subgraph of $G$. Assume that this is not the case, i.e.~$N_C(C) \subseteq N[v]$. Since $u$ and $v$ overlap it must hold that $N[u] \setminus N[v] \neq \emptyset$. The only vertices that can be adjacent to $N[u]$ but not to $N[v]$ must be in $N_C(u,u')$ since $N_C(C) \subseteq N[v]$. Therefore there exists a vertex $x \in N_C(u,u')$ that is not adjacent to $v$. For the same reason there must be a vertex $y \in N_C(w,w')$ not adjacent to $v$ because $N[w] \setminus N[v] \neq \emptyset$. The vertices of $C$ along with $v$, $x$ and $y$ form an induced $X_0$-subgraph.
\end{proof}

\begin{definition}[Restricted CA Matrix]
    Let $\lambda$ be a CA matrix. We say $\lambda$ is a restricted CA matrix if it contains an induced 4-cycle $C=(u,w,w',u')$ called witness cycle such that:
    \begin{enumerate}
        \item $N_C(u,w)$, $N_C(u',w')$ and $N_C(x)$ are empty for every $x \in C$
        \item For all $x \in N_C(C)$ it holds that $x$ overlaps with all vertices in $C$ 
    \end{enumerate} 
    \label{def:rca}
\end{definition}

Observe that the intersection matrix of every CA model that is shown in Figure~\ref*{fig:ex_nuca} is a restricted CA matrix.

\begin{table}
    \begin{tabular}{ l | l l l l | l l l l | l l | l l | l l | l l | l }
        &  \multicolumn{4}{c|}{1
        } & \multicolumn{4}{c|}{2
        } & \multicolumn{2}{c|}{3
        }  & \multicolumn{2}{c|}{4
        }  & \multicolumn{2}{c|}{5
        }  & \multicolumn{2}{c|}{6
        } & \multicolumn{1}{c}{7
        }  \\
        \hline
        $u$&$\contains$&$\overlap$&$\contains$&$\overlap$&\cellcolor{gray!25}$\disjoint$&\cellcolor{gray!25}$\disjoint$&\cellcolor{gray!25}$\disjoint$&\cellcolor{gray!25}$\disjoint$&$\overlap$&$\overlap$&$\overlap$&$\contains$&$\overlap$&$\overlap$&\cellcolor{gray!25}$\disjoint$&\cellcolor{gray!25}$\disjoint$&$\overlap$\\
        $w$&\cellcolor{gray!25}$\disjoint$&\cellcolor{gray!25}$\disjoint$&\cellcolor{gray!25}$\disjoint$&\cellcolor{gray!25}$\disjoint$&$\contains$&$\overlap$&$\contains$&$\overlap$&\cellcolor{gray!25}$\disjoint$&\cellcolor{gray!25}$\disjoint$&$\overlap$&$\overlap$&$\overlap$&$\contains$&$\overlap$&$\overlap$&$\overlap$\\
        $w'$&\cellcolor{gray!25}$\disjoint$&\cellcolor{gray!25}$\disjoint$&\cellcolor{gray!25}$\disjoint$&\cellcolor{gray!25}$\disjoint$&$\contains$&$\contains$&$\overlap$&$\overlap$&$\overlap$&$\overlap$&\cellcolor{gray!25}$\disjoint$&\cellcolor{gray!25}$\disjoint$&$\overlap$&$\overlap$&$\overlap$&$\contains$&$\overlap$\\
        $u'$&$\contains$&$\contains$&$\overlap$&$\overlap$&\cellcolor{gray!25}$\disjoint$&\cellcolor{gray!25}$\disjoint$&\cellcolor{gray!25}$\disjoint$&\cellcolor{gray!25}$\disjoint$&$\overlap$&$\contains$&$\overlap$&$\overlap$&\cellcolor{gray!25}$\disjoint$&\cellcolor{gray!25}$\disjoint$&$\overlap$&$\overlap$&$\overlap$
    \end{tabular}   
    \caption{Intersection types of restricted CA matrices with witness cycle $(u,w,w',u)$}
    \label{tab:rca}
\end{table}

\begin{fact}
    Given an intersection matrix $\lambda$, vertices $x,y_1,\dots,y_k$ of $\lambda$ and intersection types $\alpha_1,\dots,\alpha_k$, we say $x$ is an $(\alpha_1,\dots,\alpha_k)$-neighbor of $(y_1,\dots,y_k)$ if $\lambda_{x,y_i} = \alpha_i$ for all $i \in [k]$.        
    A CA matrix $\lambda$ is restricted iff $\lambda$ contains an induced 4-cycle $C=(u,w,w',u')$  such that for all vertices $x \in V(\lambda) \setminus C$ there exists a column $\overline{\alpha}$ in Table \ref{tab:rca} such that $x$ is a $\overline{\alpha}$-neighbor of $C$.  
    \label{fact:rcait}  
\end{fact}
\begin{proof}
    We use the numbers in the table headline to refer to the different columns. For example, 2.3 refers to the third column from left in the second part of the table: $(\disjoint,\contains,\overlap,\disjoint)$.
    
    ``$\Rightarrow$'': Let $\lambda$ be a restricted CA matrix with witness cycle $C=(u,w,w',u')$. We need to show for every $x \in V(\lambda) \setminus C$ there exists a column $\overline{\alpha}$ in Table \ref{tab:rca} such that $x$ is a $\overline{\alpha}$-neighbor of $C$. 
    Due to the definition of restricted CA matrices it must hold that $x$ is in (exactly) one of the following seven sets: $N_C(C), N_C(u,u'), N_C(w,w')$ or $N_C(C \setminus \{z\})$ for a $z \in C$. If $x$ is in $N_C(C)$ then $x$ overlaps with every vertex of $C$ by definition. This corresponds to the last column 7.1 of the table. If $x \in N_C(u,u')$ then $x$ is disjoint from $w$ and $w'$. In that case $x$ is an $\overline{\alpha}$-neighbor of $C$ where $\overline{\alpha}$ must be one of the four columns in part one of the table. For the same reason if $x \in N_C(w,w')$ then it is an $\overline{\alpha}$-neighbor of $C$ where $\overline{\alpha}$ corresponds to one of the two columns in the second part of the table. If $x$ is in $N_C(C \setminus \{w\})$ then $x$ is disjoint from $w$ and $x$ overlaps with both $u$ and $w'$. The intersection type between $x$ and $u'$ can be one of the following: $x$ overlaps with $u$ or $x$ is contained by $u$ or $x$ contains $u$. The first two cases are covered by the third part of the table. However, if $x$ contains $u$ then there exists no corresponding column in the table since it does not have any $\contained$-entries. This can be resolved by using the following observation: if $x$ is in $N_C(C \setminus \{w\})$ and contains $u'$ then $(u,w,w',x)$ is a witness cycle of $\lambda$ as well.  As a consequence we can assume without loss of generality that a witness cycle $C$ of $\lambda$ can be chosen such that there exists no $x \in N_C(C \setminus \{w\})$ which contains $u'$. The same argument applies to the remaining three cases $x \in N_C(C \setminus \{z\})$ with $z \in \{u,u',w'\}$.      
    
    ``$\Leftarrow$'': clear.
\end{proof}

In the remainder of this section we prove that the canonical representation problem for CA graphs is logspace-reducible to the canonical representation problem for vertex-colored restricted CA matrices. The proof outline looks as follows. First, we define a subset of uniform CA graphs, namely $\Delta$-uniform CA graphs, for which the globally invariant non-Helly triangle representability problem can be solved in logspace. Therefore the canonical representation problem for CA graphs is logspace-reducible to that of CA graphs which are not $\Delta$-uniform. This reduction follows from a slightly modified version of Lemma~\ref{lem:ginhtp}. Then we show that the neighborhood matrix of a non-$\Delta$-uniform CA graph can be converted into a vertex-colored restricted CA matrix by flipping `long' arcs. By coloring the flipped arcs the isomorphism type is preserved. 

\begin{definition}
    For a graph $G$ we define $\Delta_G$ as the following set of $\overlap$-triangles (see Definition~\ref{def:ovtriangle}).
    An $\overlap$-triangle $T$ of $G$ is in $\Delta_G$ if there exist three pairwise different vertices $u,v,w$ in $T$ such that the  following holds: 
    \begin{enumerate}
        \item $N[u] \cup N[v] \cup N[w] = V(G)$
        \item For all $z \in T$ it holds that if a vertex $x \in N_T(z)$ then $x \: \contained \: z$
        \item If there exist $u',w'$ such that $(u,w,w',u')$ is an induced 4-cycle and $v$ overlaps with $u'$ and $w'$ then $N[v] \subseteq N[u'] \cup N[w']$ 
    \end{enumerate}
    \label{def:deltag}
\end{definition}

\begin{definition}
    A CA graph $G$ is $\Delta$-uniform if $\Delta_G \cap \ovtit(G) = \emptyset$.
\end{definition}

Let us explain the intuition behind these two definitions. The set $\Delta_G$ approximates  $\ovtnht(G)$. More precisely, whenever an $\overlap$-triangle $T = \{u,v,w\}$ is in $\ovtnht(G)$ this implies that $T$ satisfies certain constraints such as for example $N[u] \cup N[v] \cup N[w] = V(G)$. The set $\Delta_G$ consists of three such constraints. Therefore if an $\overlap$-triangle is representable as non-Helly triangle it must also be in $\Delta_G$, i.e.~$\ovtnht(G) \subseteq \Delta_G$. 
The $\Delta$-uniform CA graphs can be alternatively seen as the subset of uniform CA graphs where the constraints of $\Delta_G$ suffice to characterize $\ovtnht(G)$, i.e.~$\Delta_G = \ovtnht(G)$.   

\begin{lemma}
    For every graph $G$ it holds that $\ovtnht(G) \subseteq \Delta_G$. If $G$ is a $\Delta$-uniform CA graph then $\ovtnht(G) = \Delta_G$.
    \label{lem:nht_deltag}
\end{lemma}
\begin{proof}   	
	For the first claim consider a graph $G$. If $G$ is not a CA graph then $\ovtnht(G) = \emptyset$. Therefore we can assume that $G$ is a CA graph. Given an $\overlap$-triangle $T \in \ovtnht(G)$	we show that it must be in $\Delta_G$. Let $\rho \in \mathcal{N}(G)$ be a representation such that $T = \{u,v,w\}$ is represented as non-Helly triangle in it. Since $\rho(u) \cup \rho(v) \cup \rho(w)$ covers the whole circle it follows that $N[u] \cup N[v] \cup N[w] = V(G)$, which is the first condition of Definition \ref{def:deltag}. To see that the second condition holds we consider a vertex $x \in N_T(u)$ without loss of generality. Since $x$ is not adjacent to $v$ and $w$ it holds that $\rho(x) \subseteq \mathbb{C} \setminus \left( \rho(v) \cup \rho(w) \right) $ where $\mathbb{C}$ denotes the whole circle. Since  $\mathbb{C} \setminus \left( \rho(v) \cup \rho(w) \right) \subset \rho(u)$ it follows that $\rho(x) \subset \rho(u)$. Due to the fact that $\rho$ is a normalized representation this implies that $x$ is contained by $u$. To see that the third condition of $\Delta_G$ holds let $u',w'$ be vertices such that $(u,w,w',u')$ is an induced 4-cycle of $G$. Since $T$ is represented as non-Helly triangle in $\rho$ it must hold that $\{ u',v,w' \}$ is an interval triangle in $\rho$ with $\rho(v) \subset \rho(u') \cup \rho(w')$ and therefore $N[v] \subseteq N[u'] \cup N[w']$. 
	
	For the second claim let $G$ be a $\Delta$-uniform CA graph. From the previous claim we know that $\ovtnht(G) \subseteq \Delta_G$. Since every $\overlap$-triangle must be in $\ovtnht(G) \cup \ovtit(G)$ it follows that $\Delta_G \subseteq \ovtnht(G) \cup \ovtit(G)$. The definition of $\Delta$-uniform requires $\Delta_G \cap \ovtit(G) = \emptyset$ and thus $\Delta_G \subseteq \ovtnht(G)$.  
\end{proof}

\begin{fact}
    $\Delta$-uniform CA graphs are a strict subset of uniform CA graphs.
\end{fact}
\begin{proof}
    Assume there exists a $\Delta$-uniform CA graph $G$ which is not uniform. This means there exists an $\overlap$-triangle $T \in \ovtnht(G) \cap \ovtit(G)$. Due to the previous lemma it holds that $\ovtnht(G) \subseteq \Delta_G$. This implies that $T \in \Delta_G \cap \ovtit(G)$ which contradicts that $G$ is $\Delta$-uniform. Therefore every $\Delta$-uniform CA graph is uniform.
    
	An example of a uniform CA graph that is not $\Delta$-uniform is the graph $X_4$ in Figure~\ref{fig:ex_nuca}. In the third paragraph after Theorem~\ref{thm:nuwc} we argued that $X_4$ is a uniform CA graph because it has a unique normalized representation. Assume that the black arcs of $X_4$ are labeled with $u,u',v,v',w,w'$ in the same way that the representation of $\overline{3K_2}$ is labeled in Figure~\ref{fig:ex_nuca}. To see that $X_4$ is not $\Delta$-uniform it suffices to check that the $\overlap$-triangle $\{u,v,w\}$ is in $\Delta_{X_4}$ and represented as interval triangle. 
\end{proof}

\begin{corollary}
   The globally invariant non-Helly triangle representability problem for $\Delta$-uniform CA graphs can be solved in logspace.
   \label{corol:deltag_ginhtr}
\end{corollary}
\begin{proof}
    Given a CA graph $G$ and an $\overlap$-triangle $T$ output yes iff $T \in \Delta_G$. This is correct because in the case of a $\Delta$-uniform CA graph $G$ it holds that $\Delta_G = \ovtnht(G)$ (Lemma~\ref{lem:nht_deltag}). Clearly, $\Delta_G$ is computable in logspace and an invariant.
\end{proof}

\begin{lemma}
    Let $G$ be a CA graph that is not $\Delta$-uniform. Then there exists an induced 4-cycle $C = (u,w,w',u')$ such that $N[u] \cup N[w] = N[u'] \cup N[w'] = V(G)$ and a vertex $v$ that overlaps with every vertex in $C$. 
    \label{lem:nduwc}
\end{lemma}
\begin{proof}
    The argument is essentially the same as the one made for the ``$\Rightarrow$''-direction in the proof of Theorem \ref{thm:nuwc}. The difference is that instead of the stronger assumption that $T \in \ovtnht(G)$ we only require that $T \in \Delta_G$.
    
    Since $G$ is not $\Delta$-uniform there exists an $\overlap$-triangle $T=\{u,v,w\}$ of $G$ such that $T \in \Delta_G$ and there is a representation $\rho \in \mathcal{N}(G)$ such that $T$ is represented as interval triangle in $\rho$. Furthermore, let us assume w.l.o.g.~that $\rho(v) \subset \rho(u) \cup \rho(w)$. Since $T \in \Delta_G$ it holds that $N[u] \cup N[v] \cup N[w] = V(G)$. Due to the interval representation of $T$ in $\rho$ it follows that $N[u] \cup N[w] = V(G)$. Since $u$ and $w$ do not form a circle cover it must hold that there exists a vertex $u' \in N[u] \setminus N[w]$ such that $N[u'] \setminus N[u]$ is non-empty. If $u'$ is disjoint from $v$ it follows that $u'$ must be contained by $u$ from the second condition in Definition~\ref{def:deltag} of $\Delta_G$. This cannot be the case and therefore $u' \in N_T(u,v)$. For $u'$ to have a neighbor which is not adjacent to $u$ it must hold that $\rho(u') \not\subseteq \rho(u)$. Therefore $u'$ overlaps with $u$ and $v$. Let $w' \in N[u'] \setminus N[u]$. If $w' \in N_T(w)$ then $w'$ would be contained by $w$ due to the second condition of $\Delta_G$. Again, this cannot be the case and therefore $w' \in N_T(v,w)$. From the representation $\rho$ it follows that $w$ must overlap with $u'$, $v$ and $w$. Then $C=(u,w,w',u')$ is an induced 4-cycle of $G$ such that $v$ overlaps with every vertex of $C$.  It remains to show that $N[u'] \cup N[w'] = V(G)$. Due to the third condition of $\Delta_G$ it holds that $N[v] \subseteq N[u'] \cup N[w']$. Additionally, it holds that $\rho(u) \setminus \rho(v) \subset \rho(u')$ and $\rho(w) \setminus \rho(v) \subset \rho(w')$. As a consequence $N[u'] \cup N[w'] = V(G)$.
\end{proof}

\begin{corollary}
    Canonical representations for CA graphs without induced 4-cycle can be computed in logspace.
\end{corollary}
\begin{proof}
    By Lemma \ref{lem:nduwc} the class of CA graphs without induced 4-cycle is a subset of $\Delta$-uniform CA graphs and due to Corollary \ref{corol:deltag_ginhtr} and Theorem \ref{thm:eq_problem} a canonical representation for such graphs can be computed in logspace.
\end{proof}

\begin{corollary}
    Helly CA graphs are a strict subset of $\Delta$-uniform CA graphs.
\end{corollary}
\begin{proof}
    Assume $G$ is a Helly CA graph which is not $\Delta$-uniform. Then due to Lemma \ref{lem:nduwc} there exists an induced 4-cycle $C$ and a vertex $v$ not in $C$ which overlaps with every vertex in $C$. In any normalized representation of $G$ it must hold that $v$ forms a non-Helly triangle with two vertices from $C$. This contradicts that $G$ is Helly. The graph \begin{tikzpicture}[shorten >=1pt,auto,node distance=1.2cm,
  main node/.style={circle,draw}]
  
    \newcommand*{\move}{0.3}%
       
        
    \node[circle,draw,inner sep=0,minimum size=0.1cm,fill=black] (d) at (0,0) {};
    \node[circle,draw,inner sep=0,minimum size=0.1cm,fill=black] (e) at (\move,0) {};        
    \node[circle,draw,inner sep=0,minimum size=0.1cm,fill=black] (f) at (2*\move,0) {};

    \node[circle,draw,inner sep=0,minimum size=0.1cm,fill=black] (a) at (0,\move) {};
    \node[circle,draw,inner sep=0,minimum size=0.1cm,fill=black] (b) at (\move,\move) {};        
    \node[circle,draw,inner sep=0,minimum size=0.1cm,fill=black] (c) at (2*\move,\move) {};   
    
    \path[-]
    (a) edge (d)
    (d) edge (a)
    
    (b) edge (c)
    (c) edge (b)
    
    (b) edge (d)
    (d) edge (b)
    
    (b) edge (e)
    (e) edge (b)
        
    (d) edge (e)
    (e) edge (d)    
    
    (e) edge (f)
    (f) edge (e)    
    ;

\end{tikzpicture} is a $\Delta$-uniform CA graph which is not Helly.
\end{proof}

\begin{theorem}
    The canonical representation problem for CA graphs is logspace-reducible to the canonical representation problem for vertex-colored restricted CA matrices. 
    \label{thm:rca_matrix}
\end{theorem}
\begin{proof}
    For brevity let $\mathcal{Z}$ denote the set of all CA graphs which are not $\Delta$-uniform.
    Since the globally invariant non-Helly triangle representability problem for $\Delta$-uniform CA graphs can be solved in logspace (see Corollary \ref{corol:deltag_ginhtr}) it follows from a modified version of Lemma~\ref{lem:ginhtp} that the canonical representation problem for CA graphs is logspace-reducible to the canonical representation problem for vertex-colored $\mathcal{Z}$. To see this replace  `uniform' with `$\Delta$-uniform' and `non-uniform' with `non-$\Delta$-uniform' in the statement (and proof) of Lemma~\ref{lem:ginhtp}.
    
    For a CA graph $G$ let us say a subset of vertices $X$ of $G$ is an R-flip set if $\lambda_G^{(X)}$ is a restricted CA matrix.
    To find canonical representations for $\mathcal{Z}$ we construct an invariant vertex set selector $f$ such that $f(G)$ contains at least one R-flip set for every $G \in \mathcal{Z}$. 
    Then to obtain a canonical representation for $G \in \mathcal{Z}$ let $\hat{X}$ denote the R-flip set in $f(G)$ such that $\canon(\lambda_G^{(\hat{X})},c_{\hat{X}})$ is lexicographically minimal with $c_X$ being the coloring which assigns every vertex $v \in X$ the color red and the other vertices are blue. Let $\rho$ be a canonical normalized representation for $(\lambda_G^{(\hat{X})},c_{\hat{X}})$. Then $\rho^{(\hat{X})}$ is a canonical representation for $G$. Notice, that $\rho^{(\hat{X})}$ can be computed in logspace by computing canonical representations for vertex-colored restricted CA matrices. The correctness of this approach follows from the same argument made in the proof of Theorem \ref{thm:cfsf_cr} in the flip trick section. The analogy is straightforward. The R-flip sets in this context correspond to flip sets and the invariant vertex set selector $f$ takes the place of the invariant flip set function. Given a CA graph $G$ and $X \subseteq V(G)$ it can be easily checked in logspace whether $\lambda_G^{(X)}$ is a restricted CA matrix. 
        
    For a CA graph $G$ let $C(G)$ denote the set of all ordered induced 4-cycles in $G$.
    Now, we claim that the following logspace-computable function $f$ is an invariant vertex set selector with the desired property: 
    
    $$ f(G) = \bigcup_{C \in C(G)} \big\{ \set{ x \in V(G) \setminus C  }{ \exists y \in C : x \: \contains \: y } \big\} $$        
    It is not difficult to check that $f$ is invariant. It remains to argue why $f(G)$ contains at least one R-flip set for every $G \in \mathcal{Z}$. Let $G \in \mathcal{Z}$ and $C =(u,w,w',u')$ is an induced 4-cycle in $G$ such that $N[u] \cup N[w] = N[u'] \cup N[w'] = V(G)$. The existence of such an induced 4-cycle is guaranteed by Lemma \ref{lem:nduwc}. Observe that if there exists a $u_1 \in N_C(u,w,u')$ with $u_1 \: \contains \: u$ then $C_1 = (u_1,w,w',u')$ also satisfies the previous condition $N[u_1] \cup N[w] = V(G)$. Therefore we can assume that there exists no $z \in C$ and $z_1 \in N_C(N[z] \cap C)$ such that $z_1 \: \contains \: z$. From $N[u] \cup N[w] = N[u'] \cup N[w'] = V(G)$ it immediately follows that $N_C(u,w)$, $N_C(u',w')$ and $N_C(x)$ are empty for every $x \in C$. 
    
    We prove that $\lambda^{(X)}$ is a restricted CA matrix with witness cycle $C$ where $\lambda$ is the neighborhood matrix of $G$ and $X = \set{ x \in V(G) \setminus C }{ \exists y \in C : x \: \contains \: y }$.
    Note that $X \in f(G)$ via $C$.  
    To reference the neighborhoods of $G$ (which are the same as the ones of $\lambda$) or $\lambda^{(X)}$ we write $N^G$ and $N^{\lambda^{(X)}}$ to distinguish between them. 
    First, we show that $N_C^{\lambda^{(X)}}(u,w) = \emptyset$. Assume the opposite, i.e.~there exists $x \in N_C^{\lambda^{(X)}}(u,w)$. If $x$ was not flipped, i.e.~$x \notin X$, then it also holds that $x \in N^G_C(u,w)$, which contradicts that $N^G_C(u,w)$ is empty. If $x$ was flipped, i.e.~$x \in X$, then it must be the case that $x$ contains $u'$ and $w'$ in $\lambda$. This means $N_G[u'] \cup N_G[w'] \subseteq N_G[x]$ which implies that $x$ is a universal vertex in $G$ since $N_G[u'] \cup N_G[w'] = V(G)$, contradiction. For the same reason it holds that $N_C^{\lambda^{(X)}}(u',w')$ and $N_C^{\lambda^{(X)}}(z)$ are empty for all $z \in C$. It remains to show that for all $x \in N_C^{\lambda^{(X)}}(C)$ it holds that $x$ overlaps with all vertices of $C$ in $\lambda^{(X)}$. Notice that $\lambda^{(X)}_{x,z} \in \{\overlap, \contains, \circlecover \}$ for every $z \in C$. Otherwise $x$ would not be in $N_C(C)$. We consider the following two cases: in the first one we assume that $x$ contains one vertex of $C$ in $\lambda^{(X)}$ and in the second  one we assume that $x$ forms a  circle cover with one vertex of $C$  in $\lambda^{(X)}$. We prove that neither of these cases can occur and therefore $x$ must overlap with all vertices of $C$  in $\lambda^{(X)}$. For the first case assume that w.l.o.g.~$x$ contains $u$ in $\lambda^{(X)}$ and intersects with the other vertices of $C$ in $\lambda^{(X)}$. If $x \in X$ then it was flipped. It follows that $x$ was disjoint from $u$ in $\lambda$ and therefore $x \in N_C^G(w,w',u')$. Since $x \in X$ it also must hold that $x$ contains at least one of the vertices $w,w',u'$ in $G$. It follows that $x$ contains $w'$ since it cannot contain the other two in $\lambda$. However, this contradicts our choice of $C$ which says that there exists no $w'_1 \in N_C^G(w,w',u')$ such that $w'_1$ contains $w'$ in $\lambda$. 
    If $x \notin X$ then it must hold that $x$ already contained $u$ in $\lambda$. But then $x$ should be in $X$, contradiction. For the second case assume $x$ forms a circle cover with $u$ in $\lambda^{(X)}$. If $x$ forms a circle cover with $u$ then this implies that $x$ contains $w'$ in $\lambda^{(X)}$ and therefore this reduces to the first case. We conclude that both conditions of Definition \ref{def:rca} are satisfied and hence $\lambda^{(X)}$ is a restricted CA matrix.   
\end{proof}

\section{Further Research}
Finding a polynomial-time isomorphism test for CA graphs remains an open problem. We have shown that it suffices to consider only non-uniform CA graphs for this problem. This particular class of CA graphs offers quite a lot of structure, which is caused by what we named non-uniformity witnesses. It seems plausible that such witnesses can be exploited to devise an isomorphism test. Additionally, we proved that the canonical representation problem for CA graphs is logspace-reducible to that of restricted CA matrices. The central question with regard to the flip trick is how invariant flip sets for restricted CA matrices or non-uniform CA graphs can be computed. Also, we remark that CA representations for CA graphs can be computed in logspace if flip sets for restricted CA matrices can be found in logspace. 
Another interesting problem is to extend Definition~\ref{def:deltag} of $\Delta_G$ such that it captures $\ovtnht(G)$ on uniform CA graphs, i.e.~$\Delta_G = \ovtnht(G)$ for all uniform CA graphs $G$. If this can be done in such a way that $\Delta_G$ remains an invariant and computable in logspace then everything that is said about $\Delta$-uniform CA graphs in section 5 also applies to uniform CA graphs.

\subparagraph*{Acknowledgements.}
We thank the anonymous reviewers for their insightful comments and suggestions that helped us to improve the quality of this work.

%



\bibliography{ca}


\end{document}